\documentclass[11pt, letterpaper]{article}

\usepackage{fullpage}
\usepackage{verbatim}
\usepackage{amssymb}
\usepackage{amsmath}
\usepackage{amsthm}
\usepackage{graphicx, epstopdf}
\usepackage{algorithm}
\usepackage[titletoc, title]{appendix}

\usepackage{algpseudocode}
\usepackage{algorithm}
\usepackage{capt-of, blindtext}

\usepackage{geometry}
\usepackage{apacite}

\usepackage{tikz}
\usetikzlibrary{arrows}

\newtheorem{definition}{Definition}

\newtheorem{theorem}{Theorem}

\newtheorem{assumption}{Assumption}

\newcommand{\R}{\mathbb{R}}
\newcommand{\ball}{\mathcal{B}_\delta}

\newcommand{\overdown}{\mbox{ }\begin{tikzpicture}
\draw (0,.1) --(.3,.1);
\draw [->](.3,.1) --(.45,0);
\end{tikzpicture} \mbox{ } }

\title{\LARGE \bf
Dual-mode Dynamic Window Approach to Robot Navigation with Convergence Guarantees
}

\author{Greg Droge \thanks{Email: greg.droge@usu.edu.} \thanks{Department of Electrical and Computer Engineering, Utah State University, Logan, UT 84322, USA.}
}
\date{\vspace{-5ex}}

\begin{document}


\maketitle

\begin{abstract}
In this paper, a novel, dual-mode model predictive control framework is introduced that combines the dynamic window approach to navigation with reference tracking controllers.   This adds a deliberative component to the obstacle avoidance guarantees present in the dynamic window approach as well as allow for the inclusion of complex robot models.  The proposed algorithm allows for guaranteed convergence to a goal location while navigating through an unknown environment at relatively high speeds.  The framework is applied in both simulation and hardware implementation to demonstrate the computational feasibility and the ability to cope with dynamic constraints and stability concerns.
\end{abstract}

\section{Introduction}
Vehicle motion planning in unknown environments forms an integral part of robotics.  It has been the subject of much research and is arguably a solved problem under certain conditions, e.g. \cite{Latombe1990, Arkin1998, LaVell2006, Goerzen2010}.  However, when stability becomes an issue, e.g. at high speeds, or when optimality considerations are to be taken into account, the problem is not yet solved.  Even when the environment is completely known in advance, optimal solutions can be difficult to compute as dynamic constraints, such as acceleration and motion limitations, must be considered, especially as the speed of the robot increases, e.g. \cite{Goerzen2010, Stachniss2002}.    In \cite{Goerzen2010} it was noted that analytic solutions to the optimal motion planning problem are only computable for the most simple of cases, which leads to the need for approximation algorithms for pretty much any realistic scenario.  The difficulty associated with incorporating dynamic constraints is compounded in unknown environments as a solution must be repeatedly computed to take new environmental data into account.  In this paper, we address this very issue by combining the dynamic window approach (DWA) with a fast, deliberative planner through a dual-mode model predictive control (MPC) construction.

DWA provides a direct way of incorporating dynamic constraints for fast navigation through an unknown environment, but lacks general convergence guarantees,  \cite{Fox1997, Ogren2005, Brock1999}.  The basic concept of DWA is, at each time step, to choose an arc for the robot to execute based on some predefined cost.  The cost provides a way to select more desirable arcs, such as including considerations of distance to obstacles and progression towards the goal.  Arcs are natural primitives for most wheeled mobile platforms as they can be realized by commanding constant translational and rotational velocities.  Moreover, executing arcs allows the motion of the robot to quickly be simulated into the future, which allows obstacle avoidance to be guaranteed  without a significant computational burden \cite{Fox1997}. These benefits have lead DWA to be a default ``local planner'' in the increasingly popular robot operating system's (ROS) navigation package,  \cite{ROS}.
 
The term ``local planner'' is used as DWA does not incorporate information about the connectivity of the free space when planning its action.  As such, it is known to get stuck in local minima, noted, for example, in \cite{Ogren2005, Brock1999, Stachniss2002}.  To modify the algorithm to have guarantees of convergence, \cite{Ogren2005} noted that DWA fits into the MPC framework.  In fact,  MPC is a control strategy which evaluates a cost over a certain time horizon, finds control inputs over that horizon, applies the first one (or several) of the control inputs, and then repeats the process at the next time step, e.g. \cite{Mayne2000, Primbs1999}.  Building on results from the MPC literature, \cite{Ogren2005} proved that a similar formulation to DWA can provide convergence guarantees when the motion cost is based on a navigation function, e.g. \cite{Rimon1992}.  

In a way, one can think of the navigation function in \cite{Ogren2005} as adding a deliberative component to DWA, albeit a very specific deliberative component.  The contribution of this paper is a generalization of this idea, and, similar to \cite{Ogren2005}, we capatilize on existing theory of MPC to provide guarantees of convergence, but instead of navigation functions, the deliberative component is allowed to be a generic path-planner.  Thus, a ``global planner'' finds a path, giving guarantees such as completeness, e.g. \cite{LaVell2006}, and the MPC framework takes into account the full dynamic model of the system to give guarantees of convergence to the goal location.  

The remainder of the paper will proceed as follows:  Background on the DWA algorithm and MPC are given in Section \ref{sec:Background}.  The dual-mode arc-based MPC approach is detailed in Section \ref{sec:dual_mode_MPC}.  Implementation details are given in Sections \ref{sec:Unicycle} and \ref{sec:Magellan} for an Irobot Magellan-pro.  It is shown to successfully run at 80\% of its maximum velocity, while traversing tight corridors in an unknown environment despite the robots severely limited computational resources and sluggish dynamics.  A further demonstration of the ability of the framework to deal with complicated dynamics, while maintaining the convergence guarantees, is presented in Section \ref{sec:Segway} through a simulation of an inverted pendulum robot.  The paper ends with concluding remarks in Section \ref{sec:Conclusion}.

\section{Preliminaries}
\label{sec:Background}
The backbone of the algorithm presented in this paper is an MPC technique which combines deliberate path planners with parameterized control laws to guarantee that the goal position of the planned path is reached.  This section provides the necessary background to both the choice of control laws to be used and the MPC framework to be employed in subsequent sections.

\subsection{The Dynamic Window Approach}
\label{ssec:back_DWA_MPC}
For a navigation algorithm to be successfully applied in unknown environments, it is essential to have guarantees on both obstacle avoidance and progression towards the goal location.  These guarantees need to take the dynamic constraints of the vehicle into account, especially as the speed of the vehicle increases, e.g. \cite{Goerzen2010, Stachniss2002}.  Moreover, the navigation scheme cannot consume too much computational power as the robot must perform other tasks, such as process sensor information and map the environment.  To incorporate all of these demands, we build upon the DWA algorithm presented in \cite{Fox1997}, which possesses all of the mentioned qualities, albeit without a guarantee of progression towards the goal, as noted, for example, in \cite{Ogren2005, Brock1999, Stachniss2002}.

The DWA algorithm directly deals with the dynamic constraints of the robot in order to perform fast navigation in unknown environments.  Many wheeled mobile platforms, cannot move orthogonal to the direction of motion.  This nonholonomic constraint can be expressed using the unicycle motion model, e.g. \cite{LaVell2006}:
\begin{equation}
\label{eq:unicycle}
\begin{bmatrix}
\dot{x}_1 \\
\dot{x}_2 \\
\dot{\psi}_3
\end{bmatrix} = 
\begin{bmatrix}
v \cos(\psi) \\
v \sin(\psi) \\
\omega
\end{bmatrix},
\end{equation}
where $(x_1,x_2)$ is the two dimensional position of the robot, $\psi$ is the orientation, and $v$ and $\omega$ are the translational and rotational velocities of the vehicle.  To directly deal with this constraint, the DWA algorithm utilizes it as a basis for control by considering arc-based motions, corresponding to constant $v$ and $\omega$ values.  

Commanding constant $v$ and $\omega$ values has two advantages worth mentioning.  First the $(v, \omega)$ pair can be chosen such that the transients due to the real dynamics of the vehicle can be ignored after a small window of time.  This allows for quick simulation into the future to ensure obstacles are avoided.  Second, by choosing to execute the $(v, \omega)$ pair over the simulated horizon instead of solving for a trajectory of inputs on that horizon (as is typical in MPC),  the computational burden is greatly reduced, \cite{Fox1997, Droge2012a, Park2012}.  Thus, DWA is able to guarantee obstacle avoidance while taking into account the dynamic constraints without imposing unreasonable computational burdens.  

However, in \cite{Brock1999} it was noted that the original DWA algorithm does not incorporate information about the connectivity of the free space to the goal.  Building on this idea, \cite{Ogren2005} utilized a control scheme based on incorporating navigation functions into the cost and used MPC stability techniques to ensure convergence to the goal. 
The navigation functions were chosen for a couple reasons.  They allow the cost to form a control-Lyapunov function (CLF), which is a nonlinear control method for ensuring stability, e.g. \cite{Khalil2002}, commonly used for convergence guarantees in MPC \cite{Mayne2000, Primbs1999}.  Also, the navigation functions provide an intuitive worst-case scenario: the worst the robot will do is simply follow the navigation function to the goal.

This paper builds on the idea of introducing concepts from the MPC literature to ensure convergence of a DWA-like algorithm.  By combining an arc-based controller with a reference-tracking controller, information about the connectivity of the free space to the goal can be realized with a generic path-planner instead of the need for navigation functions.  
Guarantees on the convergence to the goal location is established based on properties of the reference tracking controller as well as conditions imposed on the cost.  

\subsection{Dual-mode Model Predictive Control}
\label{ssec:dual_mode_MPC}
Dual-mode MPC is a technique that uses a stabilizing control law within an MPC framework to ensure convergence to a desired state, e.g. \cite{Mayne2000, Maniatopoulos2013}.  In this section we briefly outline a general optimal control based optimization framework for MPC and then discuss how the stabilizing control comes into play.

As mentioned in the introduction, at each iteration of an MPC algorithm a cost is minimized to find the optimal control over a certain horizon.  We denote the time at which the optimization takes place as $t_0$ and the length of the horizon as $\Delta$.  To explicitly account for the fact that MPC requires the system to simulate forward in time, we introduce a double notation for time: $x(t; t_0)$ and $u(t; t_0)$ are state and input at time $t$ as computed at time $t_0$.  Note that $x(t;t)$ and $u(t;t)$ denote the actual state and input at time $t$.  It is assumed that $x(t; t_0) \in X \subset \R^n$ and $u(t; t_0) \in U \subset \R^m$, with the dynamics denoted as $\dot{x}(t; t_0) = f(x(t; t_0),u(t; t_0))$.  The optimization problem under consideration can then be written in the following form: 
\begin{equation}
	\min_{u(\cdot; t_0)} \int_{t_0}^{t_0 + \Delta} L\bigl(x(t; t_0),u(t; t_0) \bigr)dt + \Psi(x(t_0+\Delta; t_0)),
	\label{eq:dual_mode_cost}
\end{equation} 
$$
	\mbox{s.t } \dot{x}(t; t_0) = f(x(t; t_0),u(t; t_0)), \mbox{ } x(t_0+\Delta; t_0) \in X_f,
$$
where $X_f \subset X$ is a terminal constraint set and $x(t_0; t_0)$ is known.

In dual mode MPC, it is assumed that a controller, $u = \kappa_f(x)$, exists which will render the desired equilibrium locally stable.  Without loss of generality, we assume this point to be the origin.  The way $\kappa_f(x)$ is incorporated is through an interplay between $\kappa_f(x)$, $X_f$, and the cost to give asymptotic stability of the origin.  We present the conditions on these components, given in \cite{Mayne2000}, for sake of completeness:
\begin{itemize} 
\item [B1)] $X_f$ closed 
\item [B2)] $\kappa_f(x) \in U$, $\forall x \in X_f$ 
\item [B3)] $X_f$ is positively invariant under $\dot{x} = f(x,\kappa_f(x))$
\item [B4)] $\frac{\partial \Psi}{\partial x}(x) f(x, \kappa_f(x)) + L(x, \kappa_f(x)) < 0 $ $\forall$ $x \in X_f$, $x \neq 0$
\end{itemize}
along with modest conditions on the dynamics (i.e. continuity, uniqueness of solutions, etc, e.g. \cite{Mayne2000, Chen1998}).

The method is called ``dual-mode'' MPC as an optimal control mode of operation is combined with a stabilizing control mode of operation, although the stabilizing controller is never actually used to control the system.  Intuitively, the controller is useful for a couple of reasons.  Condition B3 allows the controller to be used as a hot start for optimization as, at the current optimization time, the solution from the previous time appended with the control generated by the terminal mode will satisfy the input and terminal constraints.  Also, condition B4 states that the cost forms a CLF when using the terminal controller, making convergence somewhat expected.

\section{Dual-mode Arc-based MPC}
\label{sec:dual_mode_MPC}

The proposed dual-mode arc-based MPC algorithm builds upon DWA by using a dual-mode MPC approach to incorporate a reference tracking controller to ensure that the robot converges to some goal position while incorporating dynamic constraints.  This section develops the dual-mode approach by first presenting the algorithm, giving a convergence theorem, and then discussing how a behavior-based approach can be used as part of the optimization.

\subsection{Dual-mode Arc-baseed MPC Algorithm}
To present the dual-mode arc-based MPC algorithm, we use the behavior-based MPC framework presented in \cite{Droge2012a}.    In \cite{Droge2012a}, it is assumed that the robot will execute a given sequence of control laws, denoted as $(\kappa_0, \tau_0),(\kappa_1, \tau_1), ... , (\kappa_N, \tau_N)$, where $\tau_i$ indicates the time when the system will switch from executing the control law $\kappa_{i-1}$ to the control law $\kappa_i$.  Each control law is a function of the state and a tunable vector of parameters, written as $\kappa_i(x(t;t_0), \theta_i)$.  This allows the system dynamics to be written as $\dot{x}(t; t_0) = f(x(t;t_0), \kappa_i(x(t;t_0), \theta_i)$ for $\tau_i \leq t < \tau_{i+1}$.

In the proposed framework, we assume that the unicycle motion model in (\ref{eq:unicycle}) forms part of the state dynamics where $v$ and $\omega$ are either the inputs, as in (\ref{eq:unicycle}), or additional states of the system.  The first $N$ controllers in the sequence regulate the dynamics to desired constant velocities, with the parameter vector being the desired velocities on that interval, i.e. $\theta_i = [v_i,\mbox{ } \omega_i]^T$ for $i = 0, ..., N-1$.  The final control law is designed to track a reference trajectory, $y_d(t) \in \R^2$.  There is no need for a parameter vector, so we deviate from the original notation and write the final controller as $\kappa_N(x(t), y_d(t))$, where the time index is included to denote that $y_d(t)$ is time varying. An example of a possible trajectory is illustrated in Figure \ref{fig:cntrl_seqence} where three arc-based controllers are executed back to back with a reference tracking controller at the end. 

\begin{figure*}[!b]
\begin{center}
\includegraphics[trim=2cm 1cm 0.5cm 2cm, clip=true, width=.35\columnwidth]{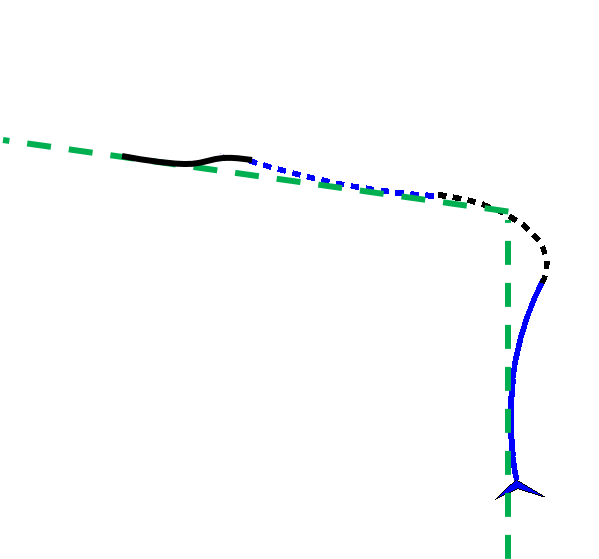}
\end{center}
\caption{This figure shows an example of a dual-mode arc-based trajectory.  The robot is shown as a triangle with a planned trajectory extending from it. The trajectory is created from three arc-based controllers appended back to back with a reference tracking controller at the end.  The different portions of the trajectory are differentiated by color and line styles.  A reference trajectory is also shown as a dashed line.} 
\label{fig:cntrl_seqence}%
\end{figure*}

The reference trajectory is produced by planning a path to the goal location, $y_{goal} \in \R^2$, and creating a continuous mapping from time to a position on the path.  In the examples presented in Sections \ref{sec:Magellan} and \ref{sec:Segway}, the path planning is done using $A^*$.  Mapping from time to position is done by respecting translational velocity constraints.  However, we note that this is merely an example and not essential to the formulation of the algorithm.

Various changes are made to the problem in (\ref{eq:dual_mode_cost}) to account for the fact that the control input is now defined in terms of the controllers, i.e. $u(t; t_0) = \kappa_i(x(t;t_0), \theta_i)$ $\tau_i \leq t < \tau_{i+1}$.  Most notably, the optimization is done over the switch times and parameters.  This replaces an infinite dimensional control trajectory optimization problem with a finite dimensional parameter optimization problem.  The instantaneous cost is modified accordingly to be a function of the parameter vectors, $\theta_i$.  Also, a key point to note is that the environment is not completely known at the time of optimization.  So the set of free or unexplored positions, $B_{free}(t_0) \subset \R^2$, is included as a term in the instantaneous cost.  The final change made is to explicitly represent the fact that the terminal cost depends upon the reference trajectory where $y(t;t_0) = h(x(t;t_0)) \in \R^2$ is the position of the robot that is expected to follow the reference trajectory.  The optimization problem can be written to incorporate these changes as:
\begin{equation}
	\min_{\tau, \theta} J(\tau, \theta) = \int_{t_0}^{t_0+\Delta} L\bigl(x(t;t_0), \theta, B_{free}(t_0)\bigr)dt  +  \Psi \bigl(y(\tau_{N+1};t_0), y_d(\tau_{N+1}) \bigr),
\label{eq:cost}
\end{equation}
$$
\mbox{s.t. } \dot{x}(t; t_0) = \begin{cases} 
f\bigl(x(t; t_0), \kappa_i(x(t; t_0), \theta_i)\bigr) & \tau_i \leq t < \tau_{i+1}\mbox{, }  i = 0, ..., N-1 \\
f\bigl(x(t; t_0), \kappa_i(x(t; t_0), y_d(t))\bigr)   & \tau_N \leq t < \tau_{N+1} 
\end{cases}
$$
where $\tau = [\tau_0, ..., \tau_{N+1}]^T$, and $\theta = [\theta_1^T, ..., \theta_{N-1}^T]^T$.  Note that we actually do not optimize with respect to the first and last elements of $\tau$, rather we fix them as $\tau_0 = t_0$ and $\tau_{N+1} = t_0 + \Delta$.    To denote that the parameters only enter the cost for the time period over which they are used, the instantaneous cost can be written as:
$$
L\bigl(x(t;t_0), \theta, B_{free}(t_0)\bigr) = \begin{cases} 
L_i \bigl(x(t;t_0), \theta_i, B_{free}(t_0)\bigr) & \tau_i \leq t < \tau_{i+1}\mbox{, }  i = 0, ..., N-1 \\
L_N \bigl(x(t;t_0), B_{free}(t_0)\bigr) & \tau_N \leq t < \tau_{N+1} 
\end{cases}.
$$
Also note that we do not have any terminal constraint on the state.  To ensure reference tracking, a terminal constraint will be re-introduced in Section \ref{ssec:convergence}.  However, to avoid computational difficulties associated with satisfying the terminal constraint during optimization, the constraint is introduced as a cost barrier in the terminal cost.

We make one final note on obstacle avoidance before stating the dual-mode arc-based MPC algorithm.  Previously unseen obstacles may render the desired reference trajectory or previously found parameters invalid due to an unforseen collision.  For the scenarios presented in Sections \ref{sec:Magellan} and \ref{sec:Segway}, it is conceivable that an obstacle avoidance controller, $\kappa_{avoid}(x)$, could consist of steering away from obstacles while slowing down as fast as possible.  The dynamic models allow angular velocities or accelerations to be controlled independent of the control of the translational velocities or accelerations.  The angular velocities or accelerations are also quite responsive.  

However, we have found the design of such an avoidance control law unnecessary.  A number of predefined parameters defining different arc-based motions can be utilized.  As will be discussed in detail in Section \ref{ssec:Optimization}, these parameters can be quickly modified to ensure collision avoidance.  To incorporate either the design of an obstacle avoidance controller or a method of quickly searching the parameter space, the following assumption is given:

\begin{assumption}{(Collision Avoidance)}
If a collision is detected at time $t_0$ to occur at time $t_c > t_0$ for either $y_d(t_c)$ or $y(t_c; t_0)$, there exists one of the following:
\begin{enumerate}
\item A control law, $\kappa_{avoid}(x)$, which will guarantee obstacle avoidance.
\item Parameters can be quickly found such that $y(t; t_0)$ is collision-free $\forall$ $t \in [t_0, \mbox{ } t_0 + \Delta]$.
\end{enumerate} 
\label{as:collisionAvoid}
\end{assumption}

The dual-mode arc-based MPC algorithm is now stated in Algorithm \ref{alg:MPCAlg}.

\begin{algorithm} 
	\caption{\textbf{Dual-mode Arc-based MPC}}
	\label{alg:MPCAlg}
\begin{enumerate}
\item \label{MPCAlg:initialize} Initialize:
	\begin{enumerate}
		\item Plan path from $y(t_0; t_0)$ to $y_{goal}$ and assign mapping from time to position to form $y_d(t)$.
		\item \label{initialize:switches} Set $\tau_0 = ... = \tau_N = t_0$ .
	\end{enumerate}
\item \label{MPCAlg:collision} If collision is detected along $y_d(t)$ or $y(t;t_0)$ for $t \in [t_0, t_0 + \Delta]$:
	\begin{enumerate}
		\item Cost barriers in terminal cost are dropped.
		\item \label{collision:evasive}  Trivial parameters are set (or $\kappa_{avoid}$ employed).
		\item New parameters are executed until new path has been planned.
		\item \label{collision:remapTraj} Assign mapping to form $y_d(t)$.		
	\end{enumerate}
\item \label{MPCAlg:parameters} Initialize Parameters $\theta$ and $\tau$:
	\begin{enumerate}
		\item \label{initialize:previous} Test previous values of $\theta$ and $\tau$.
		\item \label{initialize:behavior} Test a variety of predefined values for $\theta$ and $\tau$.
		\item \label{initialize:parameters} Choose parameters from steps \ref{initialize:previous} and \ref{initialize:behavior} which result in lowest cost .
	\end{enumerate}
\item \label{MPCAlg:Optimize} Minimize $J(\theta, \tau)$ with respect to $\theta$ and $\tau$, using parameters from \ref{initialize:parameters} as an initialization.
\item Execute control sequence for $\delta_{execute}$ seconds.
\item \label{MPCAlg:repeat} Repeat steps \ref{MPCAlg:collision} through \ref{MPCAlg:repeat} (incrementing $t_0$ by $\delta_{execute}$).
\end{enumerate}
\end{algorithm}

\subsection{Convergence}
\label{ssec:convergence}
We now show that, under Algorithm \ref{alg:MPCAlg}, the robot will converge to the goal location, $y_{goal}$.  An underlying assumption is made that the desired reference trajectory leads to the goal location in finite time, i.e. $\exists \mbox{ } t_g \mbox{ s.t. } y_d(t) = y_{goal}$ $\forall t \geq t_g$.  Convergence is then established by ensuring that Algorithm \ref{alg:MPCAlg} has certain tracking abilities for $t < t_g$ and making dual-mode MPC arguments for $t \geq t_g$.  This section discusses these two aspects of convergence and then gives a theorem about the overall convergence of the algorithm.


\subsubsection{Sufficient Tracking of Reference Trajectory}
The idea of ``sufficient tracking'' is to ensure that after step \ref{MPCAlg:Optimize} of Algorithm \ref{alg:MPCAlg}, the robot plans on getting within $\delta$ of the reference trajectory (i.e. $y(t_0+\Delta;t_0) \in \ball(y_d(t_0+\Delta))$) and collisions are avoided (i.e. $y(t; t_0) \in B_{free}(t_0)$ $\forall t \in [t_0, t_0 + \Delta]$).  To clearly express this concept, allow a solution $(\theta, \tau, x_0)$ at time $t_0$ to consist of the parameters and switch times mentioned in (\ref{eq:cost}) along with an initial condition $x(t_0; t_0) = x_0$.  The idea of ``sufficient tracking'' is encoded through the definition of an admissible solution:
\begin{definition}{(Admissible Solution)}
A solution, $(\tau, \theta, x_0)$, is said to be admissible at time $t_0$ if simulating $\dot{x}(t;t_0) = f\bigl(x(t;t_0),\kappa_i(x(t;t_0), \theta_i)\bigr)$ for $\tau_i \leq t < \tau_{i+1}$ with initial condition $x(t_0; t_0) = x_0$ results in $y(t; t_0) \in B_{free} (t_0)$ $\forall t \in [t_0, t_0 + \Delta]$ and $y(t_0+\Delta; t_0) \in \ball(y_d(t_0+\Delta))$.
\end{definition}

Similar to dual-mode MPC, conditions on both the cost and final control are employed to ensure that step \ref{MPCAlg:Optimize} of Algorithm \ref{alg:MPCAlg} is always initialized with an admissible solution and always results in an admissible solution.  These conditions are given through the following assumptions:

\begin{assumption}{(Collision Barrier)}
The instantaneous cost, $L:\R^n \times \R^M \times B_{free} \mapsto \R$, forms a cost barrier to collisions such that $L(x, \theta, B_{free})\longrightarrow \infty$ as $dist \bigl(h(x), \mathcal{B}_{free}^c \bigr) \overdown 0$, $\forall \theta$, where $dist(y,B)$ denotes the distance from point $y$ to the set $B$.
\label{as:collisionBarrier}
\end{assumption}

\begin{assumption}{(Terminal Barrier)}
The terminal cost, $\Psi:\R^2\times\R^2\mapsto\R$, forms a cost barrier around $y_d(t_0+\Delta)$ such that $\Psi(y, y_d) \longrightarrow \infty$ as $dist \bigl(y, \ball^c(y_d) \bigr) \overdown 0$. 
\label{as:terminalBarrier}
\end{assumption}

\begin{assumption}{(Trajectory Tracking)}
If $y_d(t)$ is collision free $\forall$ $t \geq t_0$ and $y(\tau_N; t_0) \in \ball (y_d(\tau_N))$, then computing $x(t; t_0)$ using the dynamics $\dot{x}(t;t_0) = f\bigl(x, \kappa_N(x(t; t_0), y_d(t))\bigr), \mbox{ } t \geq \tau_N $ will result in $y(t; t_0) \in B_{free}(t_0)$ $\forall t \geq \tau_N$ and $y(t; t_0) \in \ball(y_d(t))$ $\forall t \geq \tau_N+\delta_{execute}$.
\label{as:tracking}
\end{assumption}

Note that in Assumption \ref{as:tracking}, $\delta_{execute}$ can be used to give the final control law sufficient time to converge to a necessary region of the state-space before it is required to have excellent tracking abilities.  Together, these three assumptions allow for a statement of sufficient tracking of the reference trajectory, as given in the following theorem:
\begin{theorem}  \label{th:tracking}
Given Assumptions \ref{as:collisionBarrier}, \ref{as:terminalBarrier}, and \ref{as:tracking}, if Step \ref{MPCAlg:Optimize} of Algorithm \ref{alg:MPCAlg} is initialized with an admissible solution at time $t_s$ and $y_d(t)$ is collision free $\forall t \geq t_s,$ at each future iteration of Algorithm \ref{alg:MPCAlg}, step \ref{MPCAlg:Optimize} will produce parameters $\theta$ and $\tau$ resulting in an admissible solution.

\begin{proof}
Due to Assumptions \ref{as:collisionBarrier} and \ref{as:terminalBarrier}, the minimization of (\ref{eq:cost}) will produce an admissible solution if it is initialized with an admissible solution.  The reason being that an admissible solution will result in a finite cost while an inadmissible solution will result in an infinte cost.  As an admissible solution will result from the optimization step, at the following iteration, step \ref{initialize:previous} of Algorithm \ref{alg:MPCAlg} will produce an admissible initialization to step \ref{MPCAlg:Optimize}.  This can be seen to be the case as, after executing the solution from the previous iteration for $\delta_{execute}$, the robot will be planning on using the final control for the last $\delta_{execute}$ seconds of the new time horizon, resulting in an admissible solution due to Assumption \ref{as:tracking}. 
\end{proof}
\end{theorem}

\subsubsection{Convergence to Goal Location}
\label{ssec:dual_mode_conv}
Assuming that $y_d(t) = y_{goal}, \forall t \geq t_g$ and $B_{free}(t)$ is constant $\forall t \geq t_g$, dual-mode MPC can be employed to ensure convergence to $y_{goal}$.  The terminal region can be given as $\ball(y_{goal})$ and  
the reference tracking control, $\kappa_N(x, y_d)$, can be used as the stabilizing controller.  Together with the costs, $\kappa_N(x, y_d)$ can be designed to satisfy B1 through B4 once $y_d(t) = y_{goal}$.  With an additional assumption on the instantaneous cost, a theorem on convergence can then be state:

\begin{assumption}{(Zero Instantaneous Cost)}
The instantaneous cost is zero when $y \in \ball(y_{goal})$ and greater than or equal to zero elsewhere. 
\label{as:zeroInst}
\end{assumption}

\begin{theorem} \label{th:goal_convergence}
 Assuming $y_d(t) = y_{goal}$, $B_{free}(t)$ is constant, and Conditions B1 through B4 are satisfied using $\kappa_f(x) = \kappa_N(x, y_{goal})$, $y(t)$ will converge asymptotically to $y_{goal}$.
\end{theorem}
\begin{proof}
The development of stability in \cite{Chen1998} can be closely followed to show asymptotic convergence.  Equation (\ref{eq:cost}) can be evaluated as a discrete time candidate Lyapunov function, i.e. $V(t_0) = J(\tau, \theta, t_0)$, where the time indexing is added to denote the cost as a function of time.  
The difference in $V$ after one time step can be written as
\begin{equation}
\begin{split}
\Delta V =& V(t+\delta_{execute}) - V(t) \\
=&\int_{t + \delta_{execute}}^{t+\Delta+\delta_{execute}} L(x(s; t+\delta_{execute}), \theta, B_{free})ds + \Psi(y(t+\Delta+\delta_{execute}, y_{goal}))  \\
-&\int_{t}^{t+\Delta} L(x(s; t), \theta, B_{free})ds - \Psi(y(t+\Delta, y_{goal}) )
\end{split},
\end{equation}
where we have written $B_{free}$ without time indexing to denote that it is constant.  The integral terms can be simplified by examining several of the presented conditions.  Condition B3 states that $\kappa_N$ maintains $y \in \ball(y_{goal})$ if $y$ enters $\ball(y_{goal})$, which is by definition a property of an admissible solution once $y_d(t) = y_{goal}$.  Also, assuming no perturbations, $x(s; t) = x(s; t+\delta)$ $\forall s \geq t+\delta$ when the parameters are maintained the same.  Assumption \ref{as:zeroInst} can then be employed to state that $L(x(s; t+\delta_{execute}), \theta, B_{free}) = 0$ for $s \in [t+\Delta, \mbox{ } t+\Delta+\delta_{execute}]$. This allows the integral terms to mostly cancel out, leaving: 
\begin{equation}
\begin{split}
\Delta V &= \Psi(y(t+\Delta+\delta_{execute}, y_{goal}))  - \int_{t}^{t + \delta_{execute}} L(x(s; t), \theta, B_{free})ds - \Psi(y(t+\Delta, y_{goal}) ) 
\end{split}
\end{equation}
Again employing Assumption \ref{as:zeroInst}, which states that $L \geq 0$, along with Condition B4 which states that $\Psi$ is decreasing when using $\kappa_N$, the result can be simplified further as:
\begin{equation}
\Delta V  \leq \Psi(y(t+\Delta+\delta_{execute}, y_{goal})) -  \Psi(y(t+\Delta, y_{goal}) ) \leq 0
\end{equation}
where strict inequality holds for all $y \neq y_{goal}$.  Then, using a discrete time version of LaSalle's invariance principle, \cite{Hurt1967}, the system will converge to the set $\{ x | h(x) = y_{goal} \}$, showing asymptotic convergence to the goal location.
\end{proof}

\subsubsection{Convergence of the Dual-mode Arc-based MPC Algorithm}
The previous two sections discuss the convergence of Algorithm \ref{alg:MPCAlg} assuming that the reference trajectory is collision free.  One final assumption is made which will allow us to state a theorem on the convergence of the dual-mode arc-based MPC algorithm.

\begin{assumption}{(Known Admissible Solution)}
\label{as:known_solution}
After re-planning $y_d(t)$ in step \ref{collision:remapTraj} of Algorithm \ref{alg:MPCAlg}, an admissible solution can be found.
\end{assumption}
Due to the tracking assumption in Assumption \ref{as:tracking} and the fact that we can use fast path planners, Assumption \ref{as:known_solution} is not limiting.  An admissible solution can almost always be found by setting the switch times to equal $t_0$, as in step \ref{initialize:switches}.  However, we note that with the initialization step discussed in detail in Section \ref{ssec:Optimization}, better parameters can typically be found.  We now state the convergence theorem: 

\begin{theorem}
Given Assumptions \ref{as:collisionAvoid} through \ref{as:known_solution} and that conditions B1 through B4 hold once $y_d(t) = y_{goal}$, executing Algorithm \ref{alg:MPCAlg} will cause $y(t) = y_{goal}$ without collision where asymptotic convergence is achieved once $y_d(t)=y_{goal}$.
\end{theorem}
\begin{proof}
Assumption \ref{as:collisionAvoid} ensures that if a collision is detected due to a previously unseen obstacle, the robot can replan a path while avoiding any obstacles.  Due to Assumption \ref{as:known_solution}, Theorem \ref{th:tracking} will then hold true for the new trajectory, guiding the robot towards the goal location until either $y_d(t) = y_{goal}$ or another collision is detected.  If $y_d(t) = y_{goal}$, then Theorem \ref{th:goal_convergence} can be applied.  If a collision is detected then the process is repeated.  If the path planner is complete, eventually the robot will explore the environment enough such that a collision free path is found, if it exists.
\end{proof}

\subsection{Behavior-based Optimization}
\label{ssec:Optimization}
In \cite{Droge2012a}, gradients were derived for the behavior-based approach.  However, it is important to note that the cost to be minimized is subject to local minima caused by both the environment and the switched-time optimization.  To avoid getting stuck in undesirable local minima, a behavior-based approach capitalizing on arc-based motion primitives can be used to find a good initialization for the gradient-based optimization.  This behavior-based approach for initialization is also useful as a technique for avoiding obstacles in step \ref{collision:evasive} of Algorithm \ref{alg:MPCAlg}.


When introducing the DWA algorithm, it was noted in \cite{Fox1997} that a single arc was used due to the complexity of searching the parameter space when more than one arc is used.  While a grid-based discretized search of the parameter space, as done in \cite{Fox1997, Stachniss2002, Brock1999}, would lead to a significant increase in computation when adding additional arcs, arc-based motion can be exploited to reduce the search space.  To exploit the arc-based motion, an analytic solution to the unicycle model in (\ref{eq:unicycle}) is given in the following Theorem: 
\begin{theorem} \label{th:ConstantVelocities}
Assuming $v$ and $\omega$ are constant, the solution to (\ref{eq:unicycle}) can be expressed as 
\begin{align*}
x_1(t) &= \frac{v}{\omega} \Bigl( \sin(\psi(t)) - \sin(\psi_{30}) \Bigr) + x_{10}  &x_1(t) &= \cos(\psi_{30}) v t + x_{10} \\
x_2(t) &= \frac{v}{\omega} \Bigl( \cos(\psi_{30}) - \cos(\psi(t)) \Bigr) + x_{20} &x_2(t) &= \sin(\psi_{30}) v t + x_{20}\\
\psi(t) &= \omega t + \psi_{30} & \psi(t) &= \psi_{30}
\end{align*}
for $\omega \neq 0$ and for $\omega = 0$ on the left and right respectively.
\end{theorem}
\begin{proof}
The proof follows from uniqueness of solutions to differential equations.  Taking the time-derivative of the given equations for $\omega \neq 0$ and $\omega = 0$, respectively, will result in the original unicycle model. 
\end{proof}
Theorem \ref{th:ConstantVelocities} can be used to quickly compute collision free trajectories when a collision has been detected.  This is detailed in the following theorem:
\begin{theorem} \label{th:scaling_arcs}
Assuming the following:
\begin{enumerate}
	\item The unicycle model for the dynamics in (\ref{eq:unicycle}) is used with initial condition $x(0) = x_0$.
	\item The state trajectory, $x(t)$ is formed using velocities $(v_i, \omega_i)$ for $\tau_i \leq t < \tau_{i+1}$, where, without loss of generality, $\tau_0 = 0$.
	\item The time of first collision is given by $t_c > 0$.
\end{enumerate}
If a scaled trajectory, $x_s(t)$ with $x_s(0) = x_0$, is found by using the velocities $(s v_i, s \omega_i)$ for $\frac{\tau_i}{s} < t < \frac{\tau_{i+1}}{s}$, where $s = \alpha \frac{t_c}{\Delta}$ and $0 < \alpha < 1$.  The following holds true:
\begin{enumerate}
	\item $x_s(t) = x(s t)$
	\item $x_s(t)$ is collision free for $t \in [0, \mbox{ } \Delta]$	
\end{enumerate}
\end{theorem}
%
\begin{proof}
A proof can be obtained by applying Theorem \ref{th:ConstantVelocities} on each interval for the original and scaled variables, and seeing algebraically that the resulting scaled trajectory, $x_s(t) = x(s t)$.  Note that if $t_c$ is the time of the first collision then $x(s \Delta) = x(\alpha t_c)$ will be a point on the original trajectory before the collision as $\alpha < 1$.    Thus, $x_s(t)$ will be collision free $\forall$ $t \in [0 \mbox{ } \Delta]$.  
\end{proof}

Applying Theorem \ref{th:scaling_arcs}, a robot with unicycle dynamics can quickly avoid collisions by slowing down appropriately.  This is seen through the illustrations shown in Figure \ref{fig:Scaling}.  It shows the trajectories resulting from the arcs before and after scaling as well as the results of scaling on the cost contour map.  It illustrates that each velocity pair that produces a collision can immediately be mapped into the collision free region by scaling.  It is important to note that this method is effectively used on robots with more complicated dynamic models in Sections \ref{sec:Magellan} and \ref{sec:Segway} as they use controllers designed to quickly converge to constant velocities.

\begin{figure}[!b]
\centering
\begin{minipage}{0.68\linewidth}
\includegraphics[width=\columnwidth]{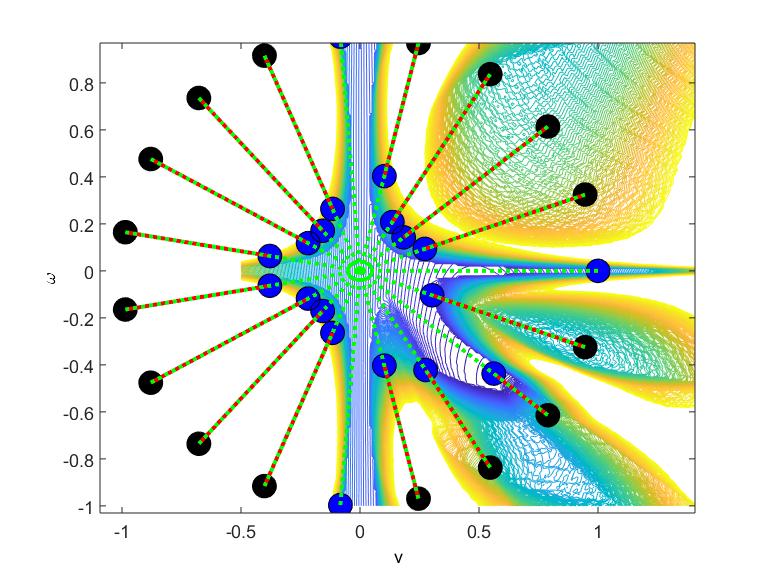}
\end{minipage}
\begin{minipage}{0.29\linewidth}
$
\begin{array}{c}
\includegraphics[width=.9\columnwidth]{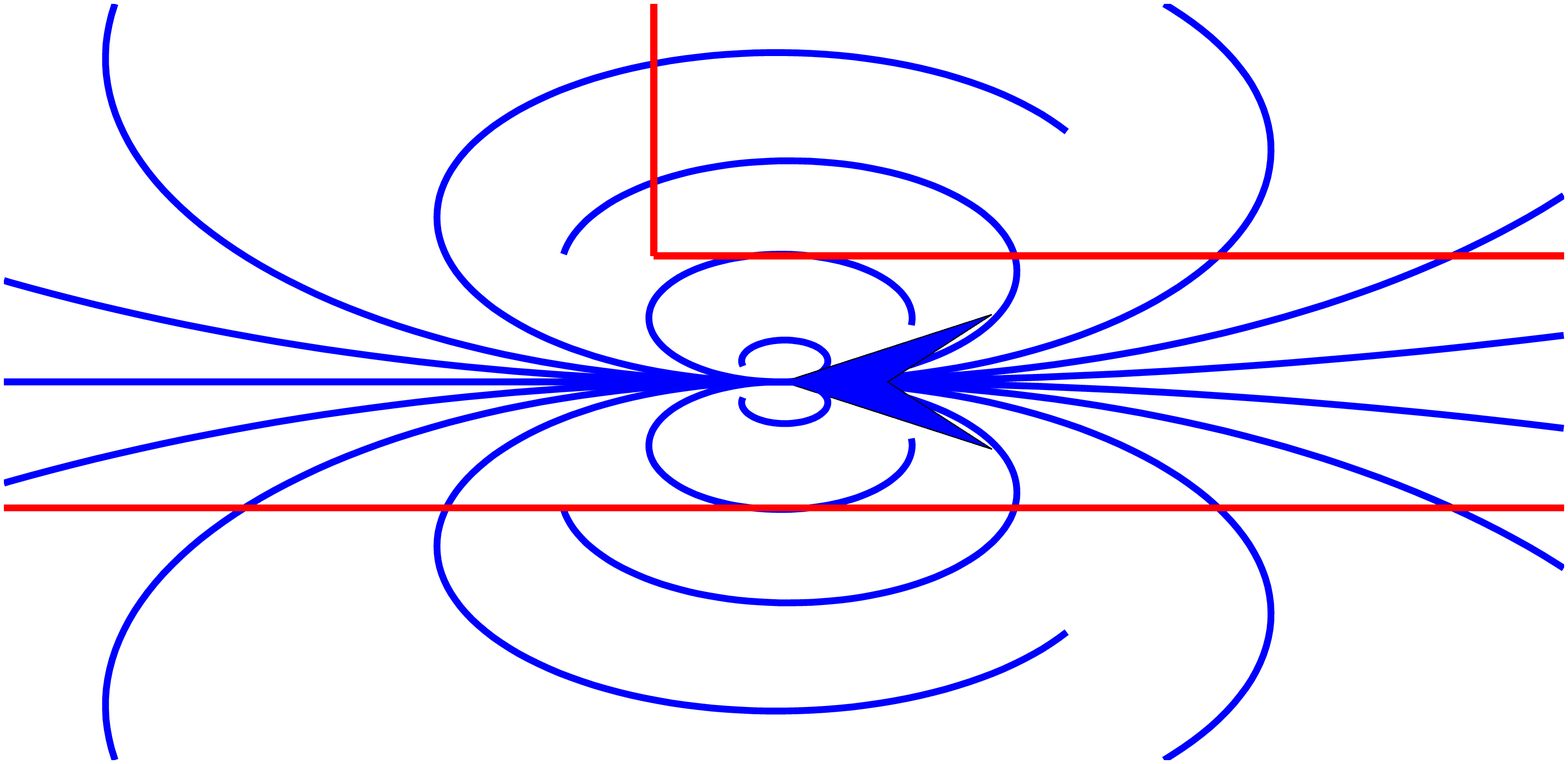} \\
\includegraphics[width=.9\columnwidth]{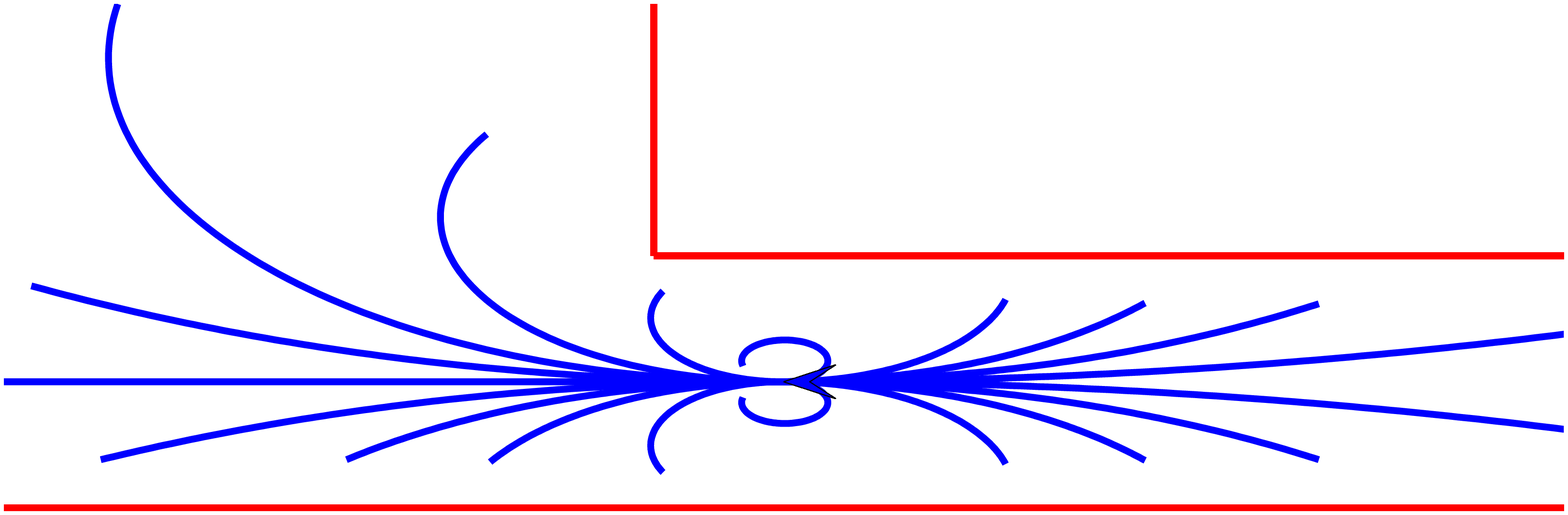} 
\end{array}$
\end{minipage}
\caption{On the left is shown a contour plot of a cost when choosing a single $(v,\omega)$ pair.  Shown on the contour plot are 20 points tested along a circle of radius 1 with a line connecting each point to a collision free pair of velocities executing the same arc, with a different speed.  On the top-right is shown the initial arcs before scaling and on the bottom-right is shown the arcs after scaling with $\alpha = .8$.  The straight lines denote walls.} 
\label{fig:Scaling}
\end{figure}

The concept of scaling can be employed to perform a quick search of the parameter space using a behavior-based approach.   Multiple arcs can be placed back-to-back to create a desired behavior.  Several examples of such behaviors that we found to be useful are illustrated in Figure \ref{fig:behaviors}.  The illustrated examples consist of the following behaviors:
\begin{itemize}
\item Single-arc behavior: all arcs are given the same velocities and $\tau_N = \tau_{N+1}$.
\item Arc-then-reference behavior: same as single-arc except $\tau_N < \tau_{N+1}$ .
\item Turn behavior: $\omega_1$, $\tau_1$, and $\tau_2$ can be designed using Theorem \ref{th:ConstantVelocities} to make a ninety-degree turn.  The remained $\omega_i$ values are set to zero.
\item Point-to-point behavior: using solutions to Theorem \ref{th:ConstantVelocities}, arcs can be computed to navigate between any two points.  
\end{itemize}
A quick search of the parameter space can then be accomplished by testing various iterations of each behavior.  For the first two behaviors, a number of arcs with different curvature can be tested (we found 20 to work quite well).  For the turn behavior, both clockwise and counter-clockwise turns can be tested as well as different lengths for the first, straight mode (we found 5 lengths to work well).  The arc and turning behaviors are illustrated further in Extension 1.  Finally, the point-to-point behavior can be used by extracting desired points from the reference trajectory.  As the behaviors are based on arcs, scaling can be used to quickly present a collision free trajectory as an option for initialization of the optimization step.

\begin{figure*}[!t]
\centering
$
\begin{array}{cccc}
\includegraphics[trim=4cm 4cm 4cm 4cm, clip=true, width=.22\columnwidth]{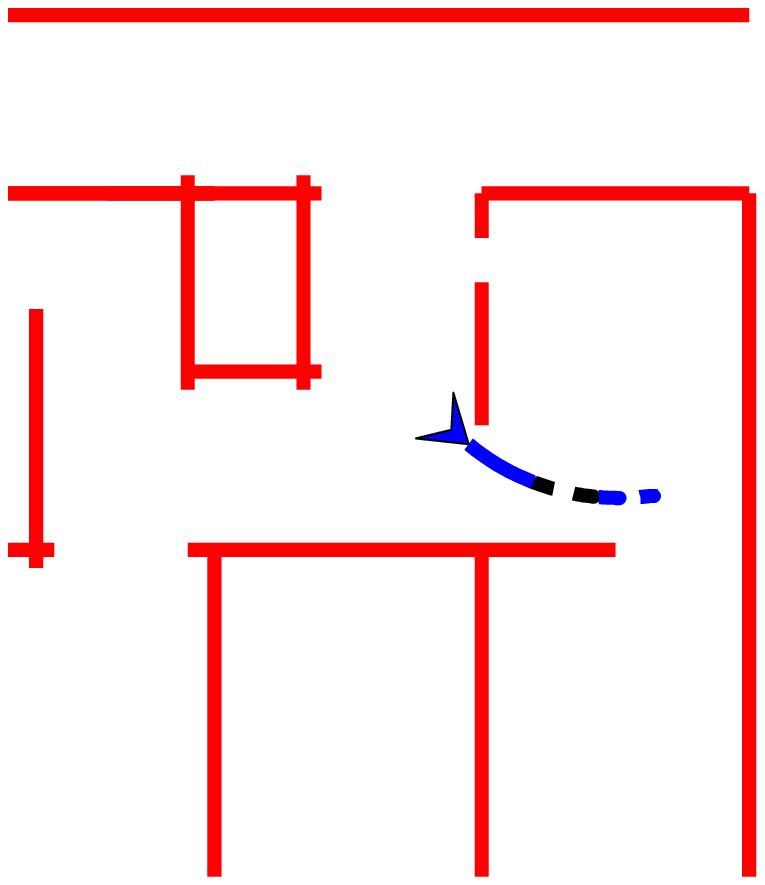} &
\includegraphics[trim=4cm 4cm 4cm 4cm, clip=true, width=.22\columnwidth]{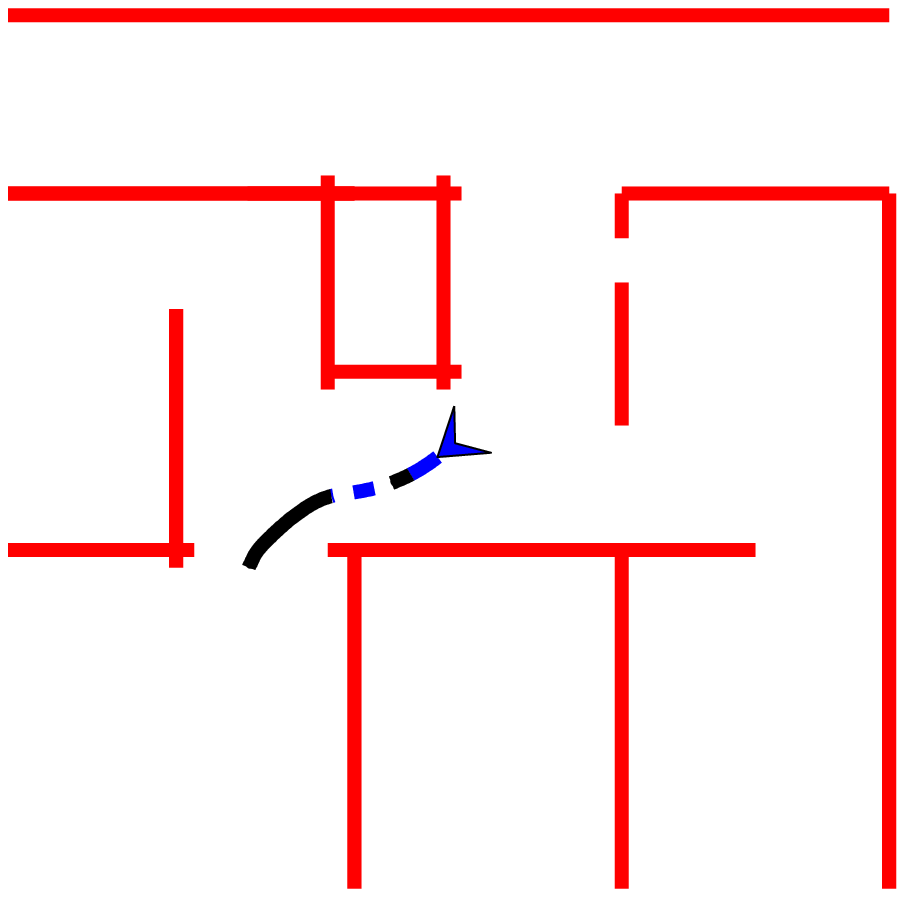} &
\includegraphics[trim=6cm 3cm 2cm 5cm, clip=true, width=.22\columnwidth]{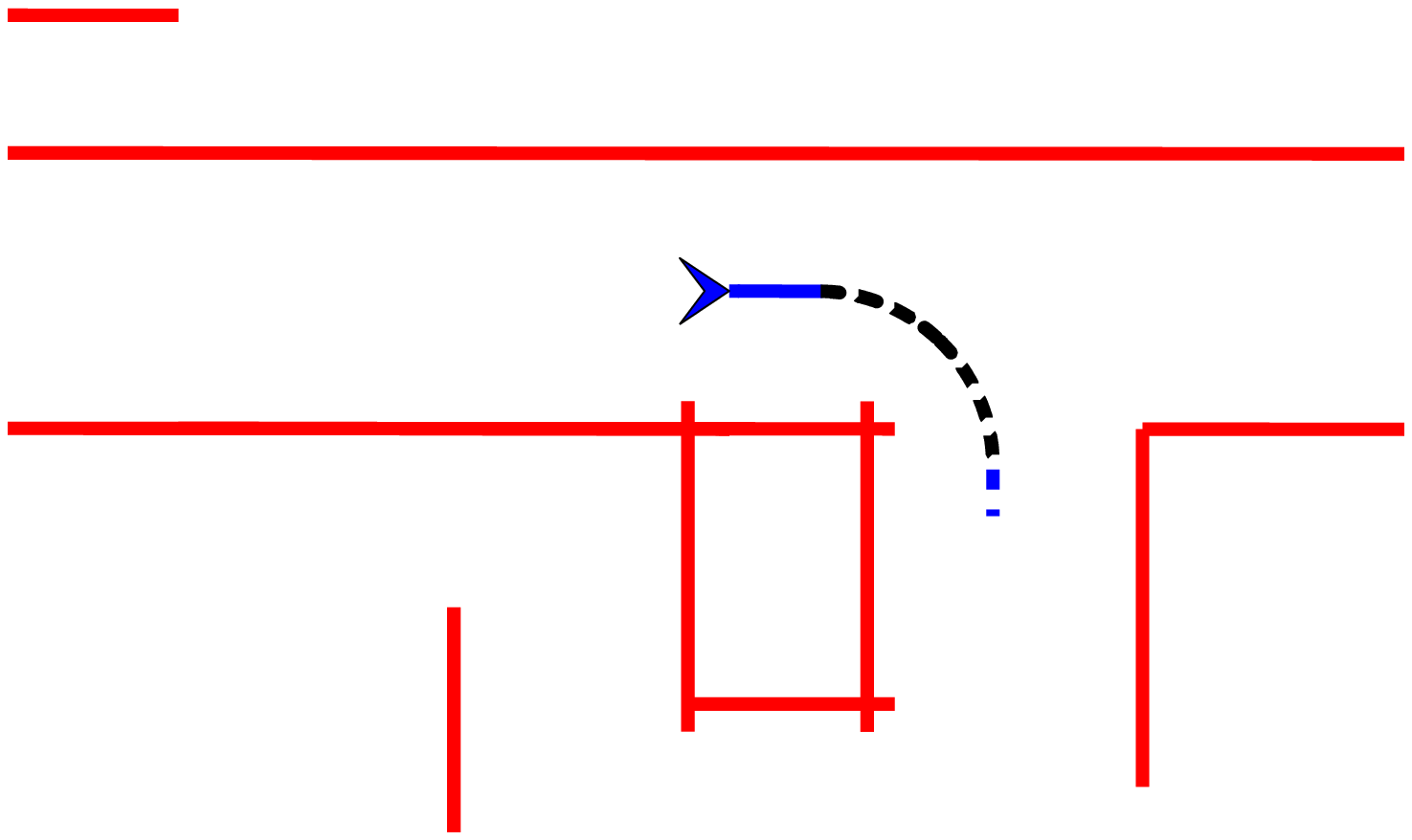} &
\includegraphics[trim=4cm 4cm 4cm 4cm, clip=true, width=.22\columnwidth]{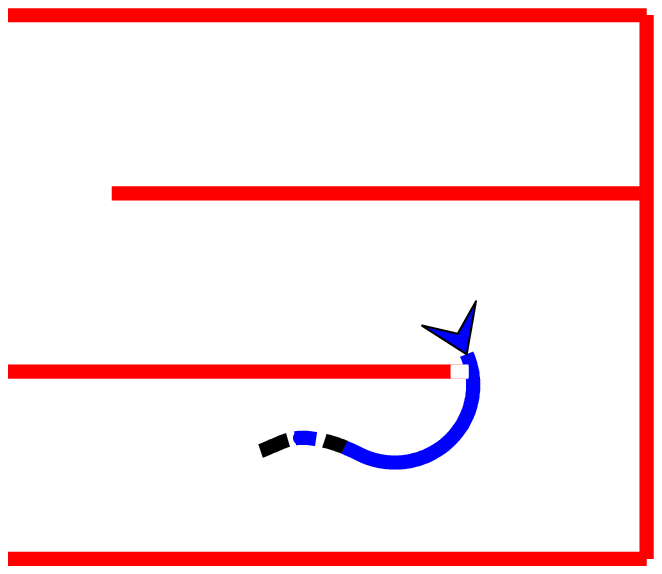} 
\end{array}$
\caption{This figure demonstrates possible behaviors based on arc-based control.  From left to right is shown a single arc, arc then reference follower, turning, and point-to-point behaviors.  In each image the robot is represented as a triangle, with the resulting trajectory extending from it.  The arcs are differentiated through color and line style, with the final reference tracking mode denoted as a solid line at the end of the trajectory (only visible in the arc-then-reference behavior). The straight lines denote walls. } 
\label{fig:behaviors}
\end{figure*}

\section{Control and Costs for Unicycle Motion Model}
\label{sec:Unicycle}
The unicycle model for motion in (\ref{eq:unicycle}) provides a basic example where considering the nonholonomic constraints of a mobile robot becomes important.  Simply planning in the position space is not sufficient as a robot with such a model is not capable of moving orthogonal to its direction of motion.  It is also a useful model to examine as control laws can be designed for robots with more complicated dynamics which, after transients die out, cause the robot to move as if its dynamic model were the unicycle model.  Results are not given in this section, but a reference trajectory tracking control law and costs are developed which are then used as a basis for the control and costs developed in Sections \ref{sec:Magellan} and \ref{sec:Segway}.

\begin{figure}[!b]
\centering
\includegraphics[width=.25\columnwidth]{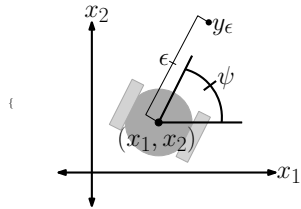}
\caption{Shown is a diagram of a two-wheeled robot with the $\epsilon$-point to be controlled.}
\label{fig:Models}
\end{figure}

\subsection{Reference Following Control}
\label{ssec:UnicycleTerminalControl}
To define a valid control law to be used as the final mode, it is necessary that it satisfy the conditions imposed in Section \ref{sec:dual_mode_MPC}.  These conditions are the tracking condition in Assumption \ref{as:tracking} as well as the dual-Mode MPC conditions in B3 and B4.  Condition B1 is satisfied by choice of goal location, and B2 is trivially satisfied as no constraints have been put on the inputs.  

To define a reference following control law satisfying these conditions, the concept of approximate $\epsilon$-diffeomorphism presented in \cite{Olfati2002} is utilized.  The idea being that, while the center of the robot has a nonholonomic constraint, a point directly in front of the robot does not.  So, instead of controlling the actual position of the robot, the final control law will control 
$$
y = \begin{bmatrix}
x_1 \\
x_2
\end{bmatrix}  + \epsilon 
\begin{bmatrix}
\cos(\psi) \\
\sin(\psi)
\end{bmatrix},
$$
where $\epsilon>0$ is some pre-defined parameter as shown in Figure \ref{fig:Models}.  Note that to simplify notation, we drop the time indexing on state, input, and reference variables throughout the remainder of the paper, except when required for sake of clarity.

A simplification of the development in \cite{Olfati2002} (i.e. controlling velocities instead of accelerations) leads to $\dot{y} = u_\epsilon$, where the inputs of the unicycle system can be solved for algebraically as
\begin{equation}
	\begin{bmatrix}
	v \\ \omega
	\end{bmatrix} =
	\begin{bmatrix}
	1 & 0 \\
	0 & \frac{1}{\epsilon}
	\end{bmatrix} 
	\begin{bmatrix}
	\cos(\psi) & \sin(\psi) \\
	-\sin(\psi) & \cos(\psi)
	\end{bmatrix} u_\epsilon.
	\label{eq:diffeomorphism}
\end{equation}
To control the unicycle to track a reference trajectory, the following controller can be used: 
\begin{equation}
\label{eq:diff_Control}
	u_\epsilon(t) = \kappa_{N_\epsilon}(y(t), y_d(t)) = \dot{y}_{d} + k_p(y_{d} - y)
\end{equation}
where $k_p > 0$ is constant and $\kappa_N(x, y_d)$ is given by combining (\ref{eq:diffeomorphism}) and (\ref{eq:diff_Control}) to obtain the commanded velocities.  A theorem is now given on the satisfaction of the tracking condition of the control law.

\begin{theorem} \label{th:diff_control_unicycle}
The control law in (\ref{eq:diff_Control}) satisfies Assumption \ref{as:tracking}.

\begin{proof}
The proof is accomplished by showing that $V(y(t)) = \frac{1}{2} || y_{d}(t) - y(t)||^2$ is a Lyapunov function.  As such, $\dot{V}$ will be shown to always be decreasing.  Therefore, if $\kappa_N(x(t), y_d(t))$ is executed starting at time $\tau_N$ and $y(\tau_N) \in \ball(y_{d}(\tau_N))$, then $y(t) \in \ball(y_{d}(t))$ $\forall t \geq \tau_N$ as the distance between $ y_{d}(t)$ and $y(t)$ will always be decreasing as time increases.

$V(t)$ can be seen to be decreasing by evaluating the time derivative:
\begin{equation}
\dot{V} 
	= (y_{d} - y)^T ( \dot{y}_{d} - \dot{y} ) = 
	(y_{d} - y)^T ( \dot{y}_{d} - u_\epsilon ) 
	= -k_p(y_{d} - y)^T (y_{d} - y)  
	< 0 \mbox{ } \forall y_{d} \neq y	
\end{equation}
\end{proof}
\end{theorem}

To show that the control law will satisfy condition B3 and B4, an assumption on the terminal cost is made and a theorem is stated.

\begin{assumption}{(Terminal Cost Convexity)}
The terminal cost is strictly convex in $y$ with a unique minimum located at $y = y_{goal}$.
\label{as:termCostConvexity}
\end{assumption}


\begin{theorem}
Given Assumptions \ref{as:zeroInst} and \ref{as:termCostConvexity}, $\Psi(y(t), y_{goal})$ satisfies B3 and B4 once $y_d(t) = y_{goal}$ when using the feedback control law $\kappa_{N}$.

\begin{proof}
In the terminal region, $y_{d}(t) = y_{goal}$ is constant which allows (\ref{eq:diff_Control}) to be reduced to $\kappa_{N_\epsilon}(y(t)) = k_p(y_{goal} - y)$.  The proof of Theorem \ref{th:diff_control_unicycle} shows that $\ball(y_{goal})$ is invariant under $\kappa_{N_\epsilon}$, satisfying B3.  To evaluate B4, $\Psi(y(t))$ needs to be evaluated as a candidate control-Lyapunov function under the control $\kappa_{N_\epsilon}$.  The time derivative can be evaluated as 
\begin{equation}
\dot{\Psi} 
	= \frac{\partial \Psi}{\partial y}(y) \dot{y} 
	= k_p\frac{\partial \Psi}{\partial y}(y)(y_{d} - y) 
	< k_p( \Psi(y) + \frac{\partial \Psi}{\partial y}(y_{d} - y) )
\end{equation}
as $\Psi(y) \geq 0  $ with strict inequality $\forall y \neq y_{d}$ due to Assumption \ref{as:termCostConvexity}.  Using the global underestimator property of convex functions, we can further state:
\begin{equation}
\begin{split}
	\dot{\Psi} 
	&< k_p( \Psi(y_{d})  ) = 0
\end{split}
\end{equation}
Therefore, $\dot{\Psi} < 0$ $\forall y \neq y_{d}$ and B4 is satisfied.
\end{proof}
\end{theorem}

\subsection{Cost Definition}
\label{ssec:UnicycleCosts}
To define valid costs for the dual-mode arc-based MPC algorithm, the costs must satisfy several conditions in terms of the $\epsilon$-point, $y$.  The terminal cost must satisfy two conditions:
\begin{enumerate}
\item Assumption \ref{as:terminalBarrier} gives a condition of a barrier around the reference position.
\item Assumption \ref{as:termCostConvexity} and B4 give a condition on the convexity.
\end{enumerate}
Similarly, the instantaneous costs must satisfy two conditions:
\begin{enumerate}
\item Assumption \ref{as:collisionBarrier} requires the instantaneous cost to form a barrier to collisions.
\item Assumption \ref{as:zeroInst} requires the instantaneous cost to approach zero as $y \longrightarrow \ball (y_{goal})$.
\end{enumerate}

The terminal cost is defined as:
\begin{equation}
	\Psi(x(t), y_d(t)) = \frac{\rho_4}{2}||y(t) - y_{d}(t)||^2 + \rho_5(-\log(\frac{dist(y(t), \ball(y_{d}(t))}{\delta} )  ).
	\label{eq:term_cost_unicycle}
\end{equation}
As both terms in  $\Psi(x(t), y_d(t))$ are strictly convex with minimum at $y = y_{d}$, Assumptions \ref{as:termCostConvexity} and B4 are satisfied.  Moreover, Assumption \ref{as:terminalBarrier} is satisfied as $\Psi$ is defined as a log-barrier.  Note that while the first term seems to be redundant, it is used to help find parameters in step \ref{MPCAlg:collision} of Algorithm \ref{alg:MPCAlg} when the terminal barrier cost is removed.

The instantaneous cost is used to ensure obstacle avoidance while moving at a desirable speed, $v_d$.  To represent the obstacles we use a grid-based map where $N_{obs}$ is the number of occupied grid points within a radius of $d_{max}$ of the robot.  It is assumed that the occupied grid points are available at time $t$ as $o_1(t), ..., o_{N_{obs}}(t)$, where $o_i(t) \in \R^2$.  This allows the cost to be written as 
\begin{equation}
\begin{split}
L(x(t), \theta_i, B_{free}(t_0)) =&  L_{goal}(||y(t)-y_{goal_\epsilon}||)  \bigl( \frac{\rho_1}{2} (v_d - v(t))^2 + \frac{\rho_2}{2}\omega(t)^2 \bigr) \\
&+ \frac{\rho_3}{2} \sum_i^{N_{obs}} L_{avoid}( || y - o_i(t) || )
\end{split}
\label{eq:inst_cost_unicycle}
\end{equation}
where
$$
	L_{goal}(d) = \begin{cases}
	\frac{2}{1+e^{-a(d - \delta)}} - 1 & d \geq \delta \\
	0 & d < \delta
	\end{cases}	
$$
$$
	L_{avoid}(d) = \begin{cases}
	-a \log(d - d_{min}) + a \log(d_{max}-d_{min})  & d_{min} < d \leq d_{max} \\
	\infty & d \leq d_{min} \\
	0 & d > d_{max} 	 
	\end{cases}
$$
$\rho_i > 0$ is a weight, and $d_{min}$ is the minimum allowed distance from an occupied grid cell to the robot.

The term $L_{goal}$ is provided to smoothly reduce the influence of the instantaneous cost around the goal to satisfy Assumption \ref{as:zeroInst} for the first two terms in the cost.  Note, to completely satisfy these two assumptions, the goal position needs to be defined such that the goal is at least $d_{max} + \delta$ from the nearest obstacle so $L_{avoid}$ is also zero in the terminal region.  By inspection, the log-barrier cost $L_{avoid}$ satisfies Assumption \ref{as:collisionBarrier}.

\section{Case Study: Magellan Robot}
\label{sec:Magellan}
The unicycle motion model and control presented in Section \ref{sec:Unicycle} can be used when the convergence time to the desired velocities is negligible.  However, it is often the case that dynamic limitations cannot be ignored, especially as the speed of the robot increases.  The Irobot Magellan-Pro, shown in Figure \ref{fig:mag_shot}, presents an excellent example of such a scenario.  We have observed the maximum translational acceleration to be approximately $.1 \frac{m}{s}$, which becomes a significant factor in avoiding obstacles when traveling near the maximum velocity of $1 \frac{m}{s}$.

\begin{figure*}[!b]
\centering
\includegraphics[width=.22\columnwidth]{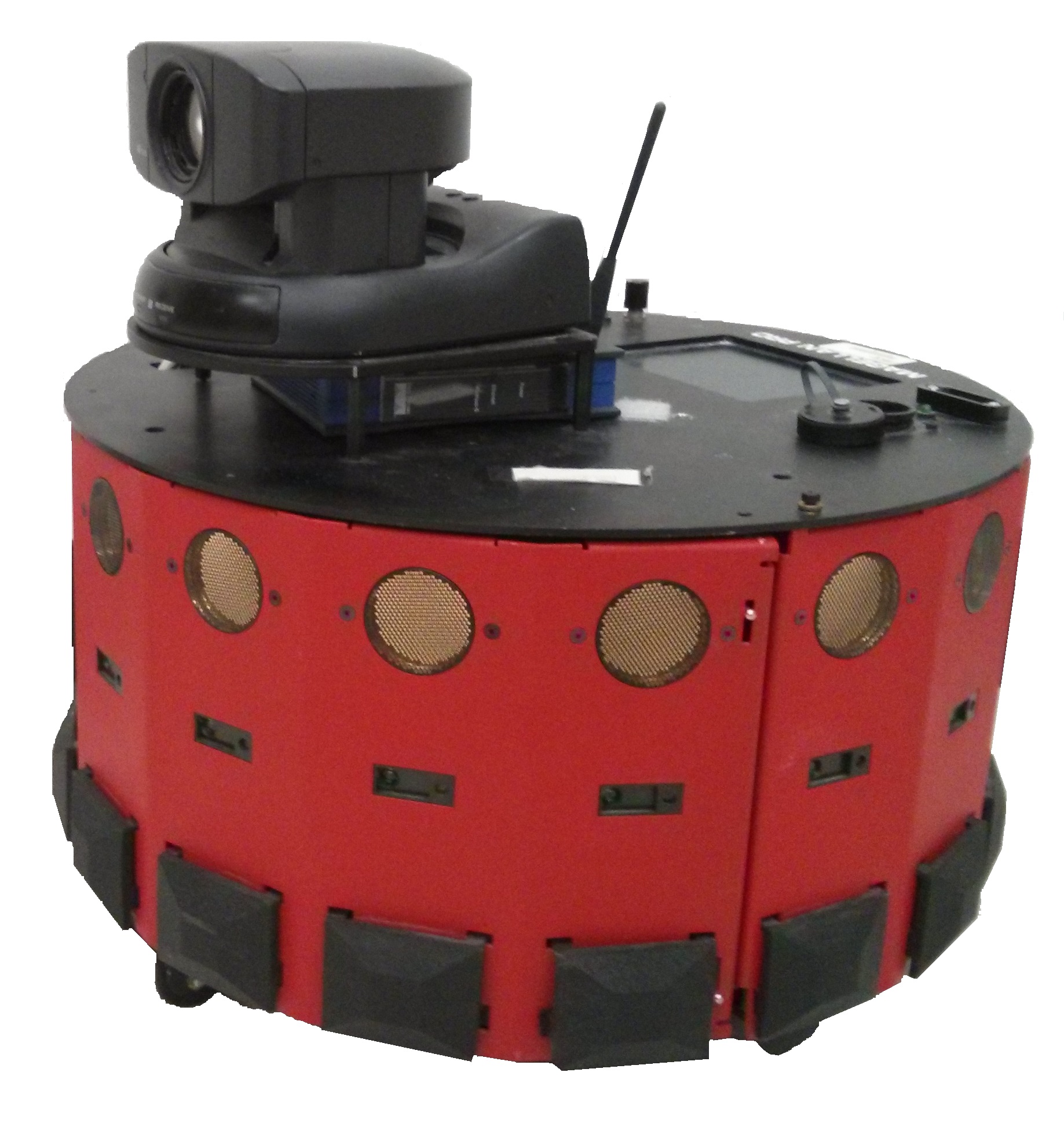} 
\caption{Shown is a picture of the IRobot Magellan-pro.} 
\label{fig:mag_shot}
\end{figure*}

The Magellan is also a prime candidate to demonstrate the dual-mode arc-based MPC algorithm due to its relatively limited computational capabilities.  It has a Pentium III 850 Mhz processor, which is limited when compared to the state of the art.  To develop the dual-mode arc-based MPC algorithm for the Magellan, this section begins with an explanation of the dynamic model to be used, develops a control law for reference tracking, and ends with experimental results.

\subsection{First Order Model}
\label{ssec:Mag_model}
While the kinematic constraints of the Magellan can be written as the unicycle model in (\ref{eq:unicycle}), the velocities cannot be instantaneously realized.  The Magellan-pro robot has an on-board control which accepts velocity commands and limits accelerations.  Thus, similar to \cite{Krajnik2011}, which modeled a UAV with a similar on-board control, we utilize a first order model on the velocity inputs written as follows:
\begin{equation}
	f(x) = 
	\begin{bmatrix}
	\dot{x}_1 \\
	\dot{x}_2 \\
	\dot{\psi} \\
	\dot{v} \\
	\dot{\omega}
	\end{bmatrix} =
	\begin{bmatrix}
	v \cos(\psi) \\ 
	v \sin(\psi) \\
	\omega \\
	a_1 v + b_1 u_1 \\
	a_2 \omega + b_2 u_2
	\end{bmatrix},
\label{eq:Magellan}
\end{equation}
$$
\mbox{s.t. } |\dot{v}| \leq a_{max}, |\dot{\omega}| \leq \alpha_{max}
$$
where $a_i, b_i, a_{max}$, and $\alpha_{max}$ are constants. 

\subsection{Reference Following Control}
\label{ssec:Mag_ctrl}
To use the dual-mode arc-based MPC algorithm, two control laws must be defined:  a control law to regulate the system to desired, constant velocities, and a controller to be used for reference tracking.  The control law used to converge to desired, constant velocities is formed using an LQ approach, e.g. \cite{Bryson1975}.  As this is a highly developed control methodology, the details are not presented.  Rather, we focus our attention on the reference tracking control as it must satisfy the conditions imposed in Section \ref{sec:dual_mode_MPC}.  These conditions are given as the tracking condition in Assumption \ref{as:tracking} as well as the dual-Mode MPC conditions given in B2, B3, and B4.  Again, condition B1 is satisfied by choice of goal location.  

To form a control law, we first ignore the acceleration constraints and control the velocities and then examine the control law to see when it can be employed.  Let $e_{\bar{v}}$ be the error between $\bar{v} = [v, \omega]^T$ and the desired velocities, denoted as $\bar{v}_d$.  Let the control be defined as follows
\begin{equation}
	u_{\bar{v}} = \begin{bmatrix}
	u_1 & u_2
\end{bmatrix}^T =	 B^{-1} \bigl( \dot{\bar{v}}_d - A \bar{v}_d \bigr) - k_v e_{\bar{v}} ,
\label{eq:MagellanCtrl1}
\end{equation} 
\begin{equation}
	\bar{v}_d = \begin{bmatrix}
	1 & 0 \\
	0 & \frac{1}{\epsilon} \\
	\end{bmatrix}
	\begin{bmatrix}
	\cos(\psi) & \sin(\psi) \\
	-\sin(\psi) & \cos(\psi)
	\end{bmatrix}
	u_\epsilon,
\label{eq:MagellanCtrl2}
\end{equation}
$$
A = \begin{bmatrix}
a_1 & 0 \\
0 & a_2 
\end{bmatrix}
, B = \begin{bmatrix}
b_1 & 0 \\
0 & b_2 
\end{bmatrix},
$$
\begin{equation}
 	\dot{\bar{v}}_d = \begin{bmatrix}
 	-\omega \sin(\psi) & \omega \cos(\psi) \\
 	-\frac{\omega \cos(\psi)}{\epsilon} &  	-\frac{\omega \sin(\psi)}{\epsilon}
 	\end{bmatrix} u_\epsilon + \begin{bmatrix}
 	\cos(\psi) & \sin(\psi) \\
	-\frac{\sin(\psi)}{\epsilon} & \frac{\cos(\psi)}{\epsilon}
 	\end{bmatrix} \dot{u}_\epsilon,
 	\label{eq:Mag_vbar}
\end{equation}
where $u_\epsilon$ is defined as in (\ref{eq:diff_Control}), and
$$
	\dot{u}_\epsilon = \ddot{x}_{d_\epsilon} + k(\dot{x}_{d_\epsilon} - y).
$$

The time derivative of $e_{\bar{v}}$ can then be written as $\dot{e}_{\bar{v}} = (A - B k_v)e_{\bar{v}} = \tilde{A}e_{\bar{v}}$ where $k_v$ is a matrix of gains used for feedback.  Thus, the error in velocities is exponentially stable to zero.  As the desired velocities converge, the system behaves as the unicycle model of Section \ref{sec:Unicycle}.  To ensure that the required conditions are met, two additional assumptions are made:

\begin{assumption}
\label{as:limitedAcc_RefTraj}
	The desired reference trajectory accelerations must be bounded below the acceleration limitations of the Magellan-pro robot.
\end{assumption}
\begin{assumption}
\label{as:proximit_vel_Mag}
	The velocities at time $\tau_N$ are sufficiently close to the desired velocities to allow convergence within $\delta_{execute}$ seconds .
\end{assumption}

Assumption \ref{as:limitedAcc_RefTraj} corresponds to a condition that the time mapping must satisfy the translational acceleration constraint and a condition on the curvature of the path must satisfy the rotational acceleration constraint.  Assumption \ref{as:proximit_vel_Mag} presents a relationship between the velocities at time $\tau_N$, $\delta_{execute}$, and the convergence rate of the control, directly affected by choice of $k_v$ in (\ref{eq:MagellanCtrl1}).  Assumption \ref{as:proximit_vel_Mag} could be guaranteed through the introduction of additional cost barriers on the velocities.  In practice, we have seen this to be unnecessary by choosing $v_d = ||\dot{y}_{d}(\tau_N)||$ for the final arc in the initialization step discussed in Section \ref{ssec:Optimization} and noting that $\omega$ converges quite quickly.


\subsection{Results}
\label{ssec:Mag_results}
To present the results for the implementation of the dual-mode arc-based MPC algorithm on the Magellan Pro robot, we perform a mix of simulation and actual implementation, where the desired velocity was set at .9 $\frac{m}{s}$, 90\% of the Magellan's top speed. A simulated environment was projected onto the floor and the robot navigated through the environment as shown in Figure \ref{fig:Mag_environement} and Extension 1.  To allow for high speeds, the environment continuously switched between two predefined environments.   Along with a simulated environment, a simulated laser scanner is used as a sensor.  The laser scanner takes 100 measurements ranging from $\psi - \frac{\pi}{2}$ to $\psi + \frac{\pi}{2}$.  The robot is interfaced to a mapping algorithm through the ROS navigation package, using an $A^*$ planning algorithm to find desired paths and the dual-mode arc-based MPC algorithm to follow the planned path.

To perform optimization, the techniques from Section \ref{ssec:Optimization} were employed.  The NLopt optimization library, \cite{NLOPT2}, together with the gradients in \cite{Droge2012a}, were used to perform the gradient-based portion of the optimization.  To test the efficacy of the different portions of the algorithm, several trials were run.  First, the tracking control was run by itself, then the arc-based algorithm was performed without the gradient descent, and finally the gradient descent was incorporated and the remaining trials consisted of varying the alloted time for gradient descent (a parameter available in the NLopt library).  

It is important to note that the tracking controller resulted in multiple collisions when running near the maximum velocity of the Magellan.  This is quite possibly due to modeling errors, which are amplified at the top speeds.  This was not a hindrance to the other trials as the tracking controller was never actually executed, due to the behavior-based initialization step to the optimization.  Thus, for a fair comparison, the results for pure tracking control were obtained from simulation.  

\begin{figure*}[!t]
\centering
$
\begin{array}{cccc}
\includegraphics[width=.22\columnwidth]{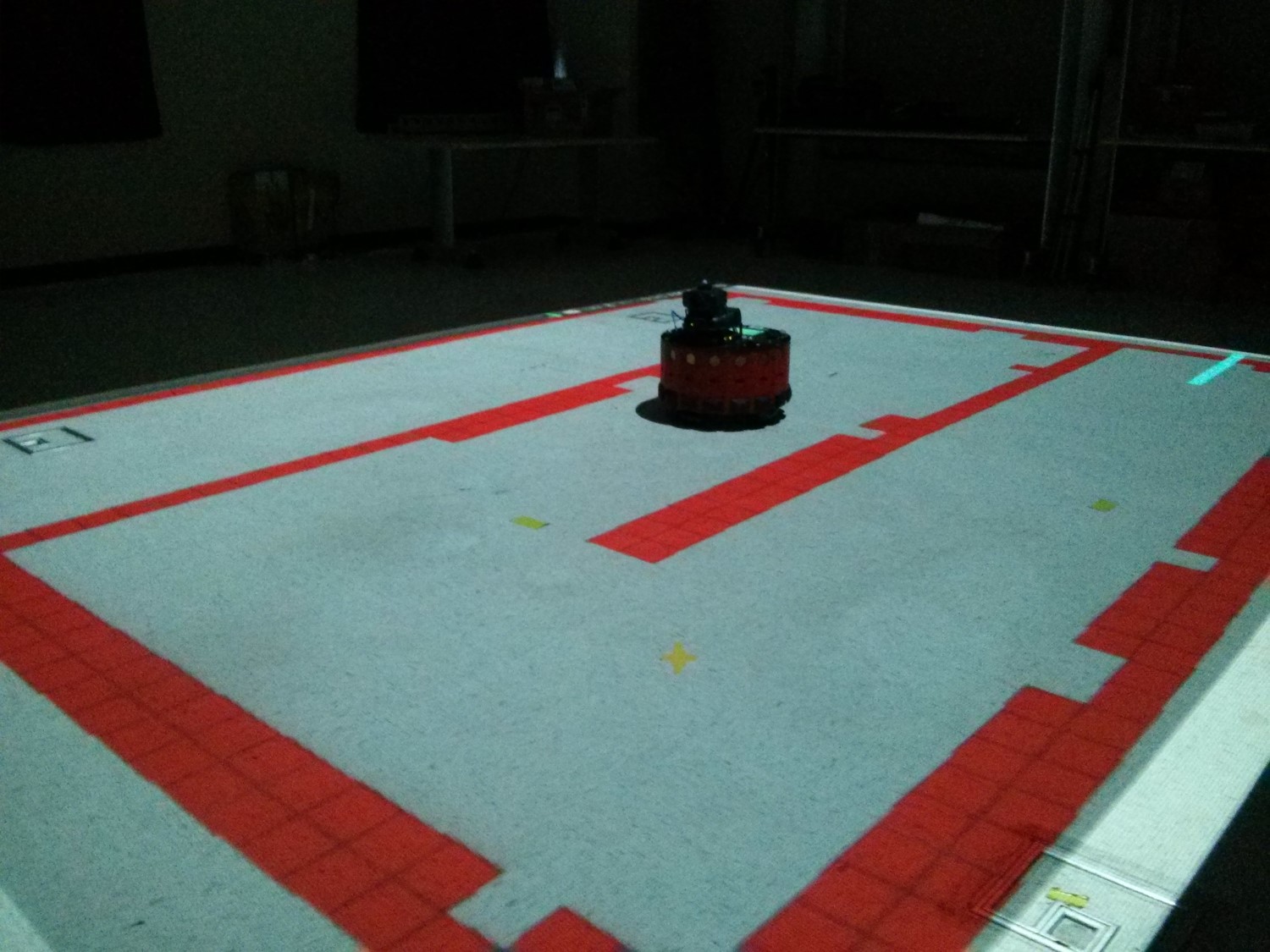} &  	
\includegraphics[width=.22\columnwidth]{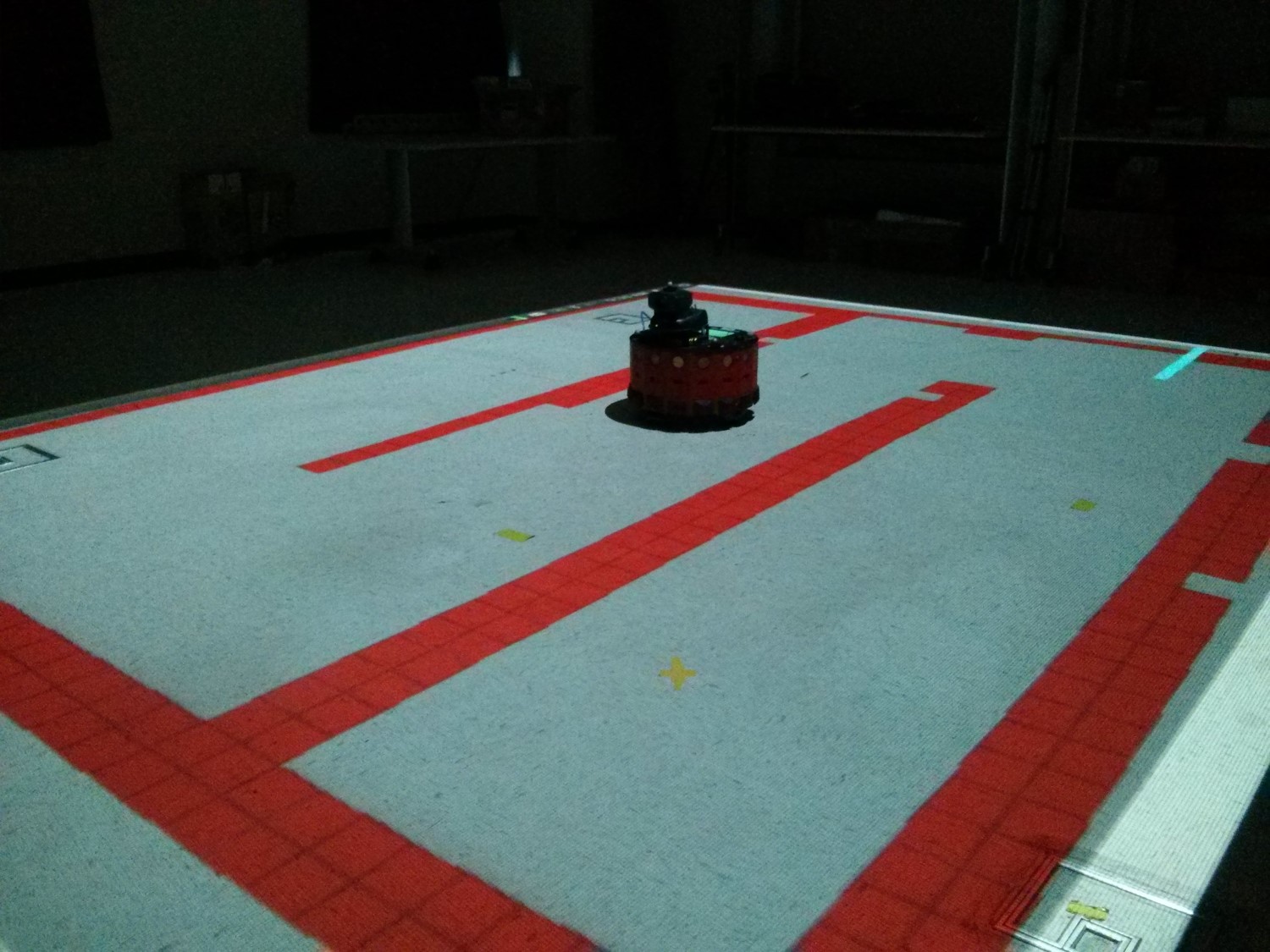} &
\includegraphics[width=.22\columnwidth]{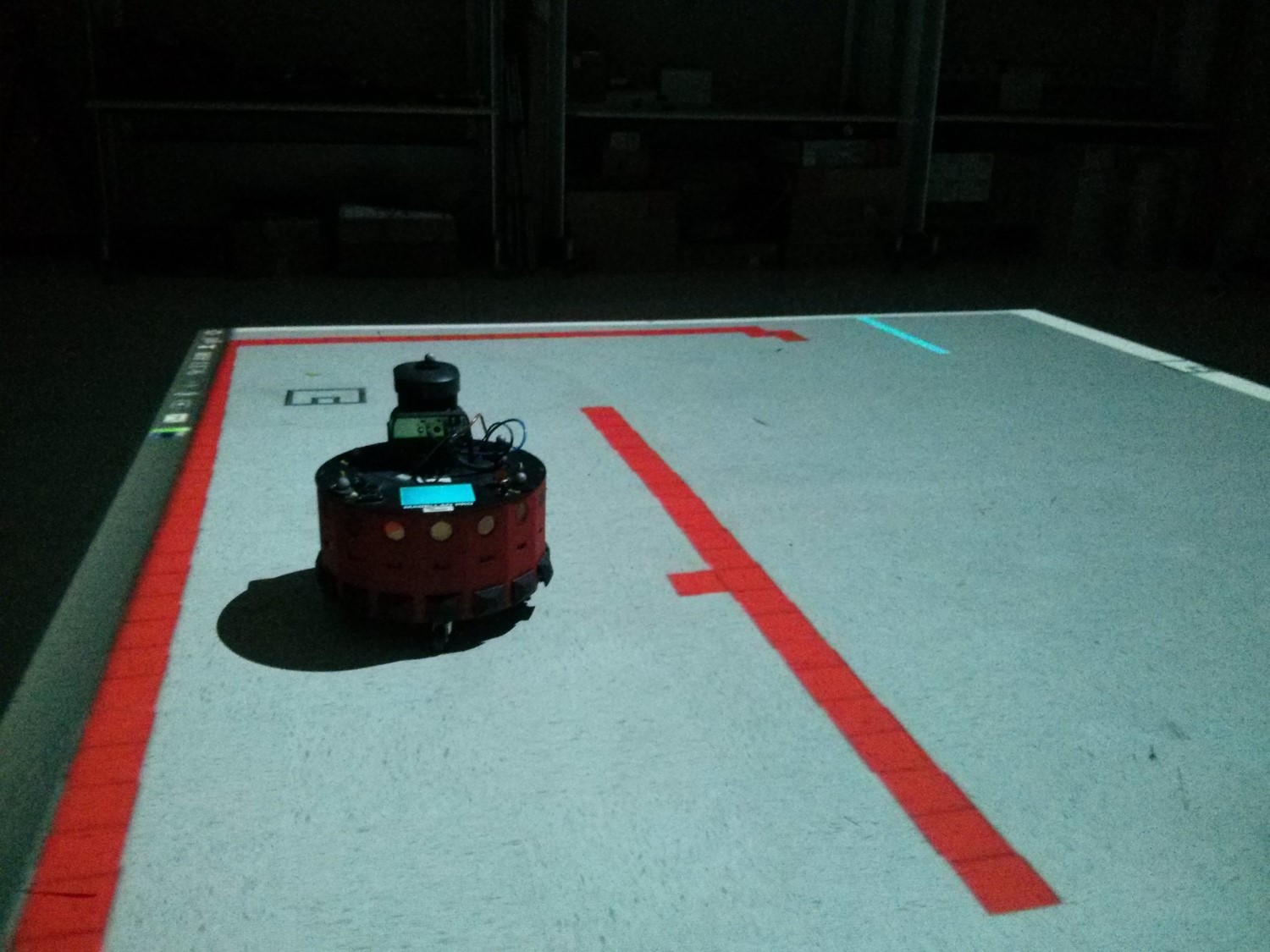} &
\includegraphics[width=.22\columnwidth]{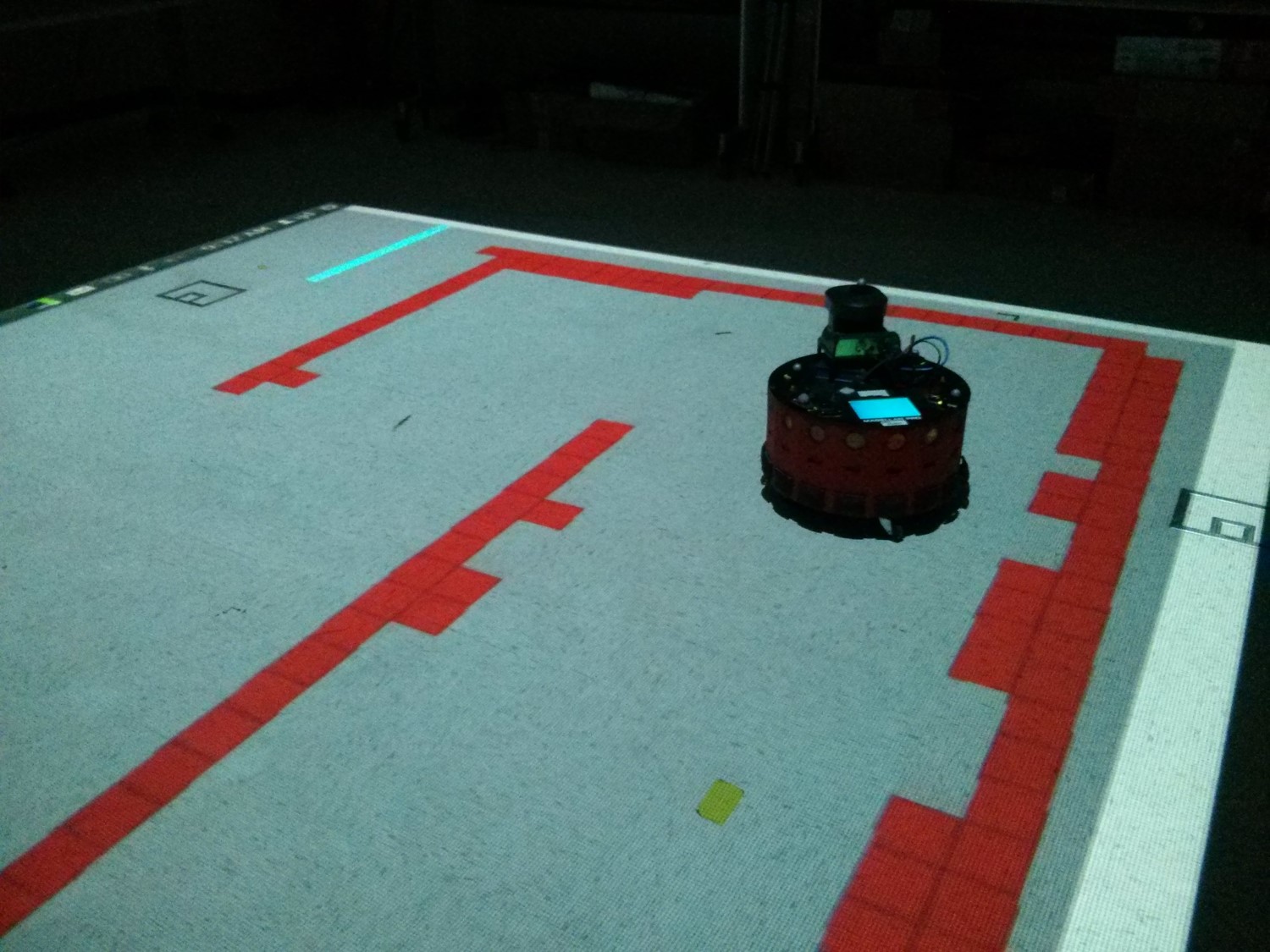} 
\end{array}$
\caption{The left two images show the two environments completely mapped with the Magellan in the center.  The right two images show the Magellan executing Algorithm \ref{alg:MPCAlg}.  As the Magellan approaches the goal, the map is erased so that the Magellan can traverse the next environment as if it were unknown.} 
\label{fig:Mag_environement}
\end{figure*}

\begin{figure*}[!t]
\centering
$
\begin{array}{ccc}
\includegraphics[trim=.5cm 1cm 6cm 0cm, clip=true, width=.3\columnwidth]{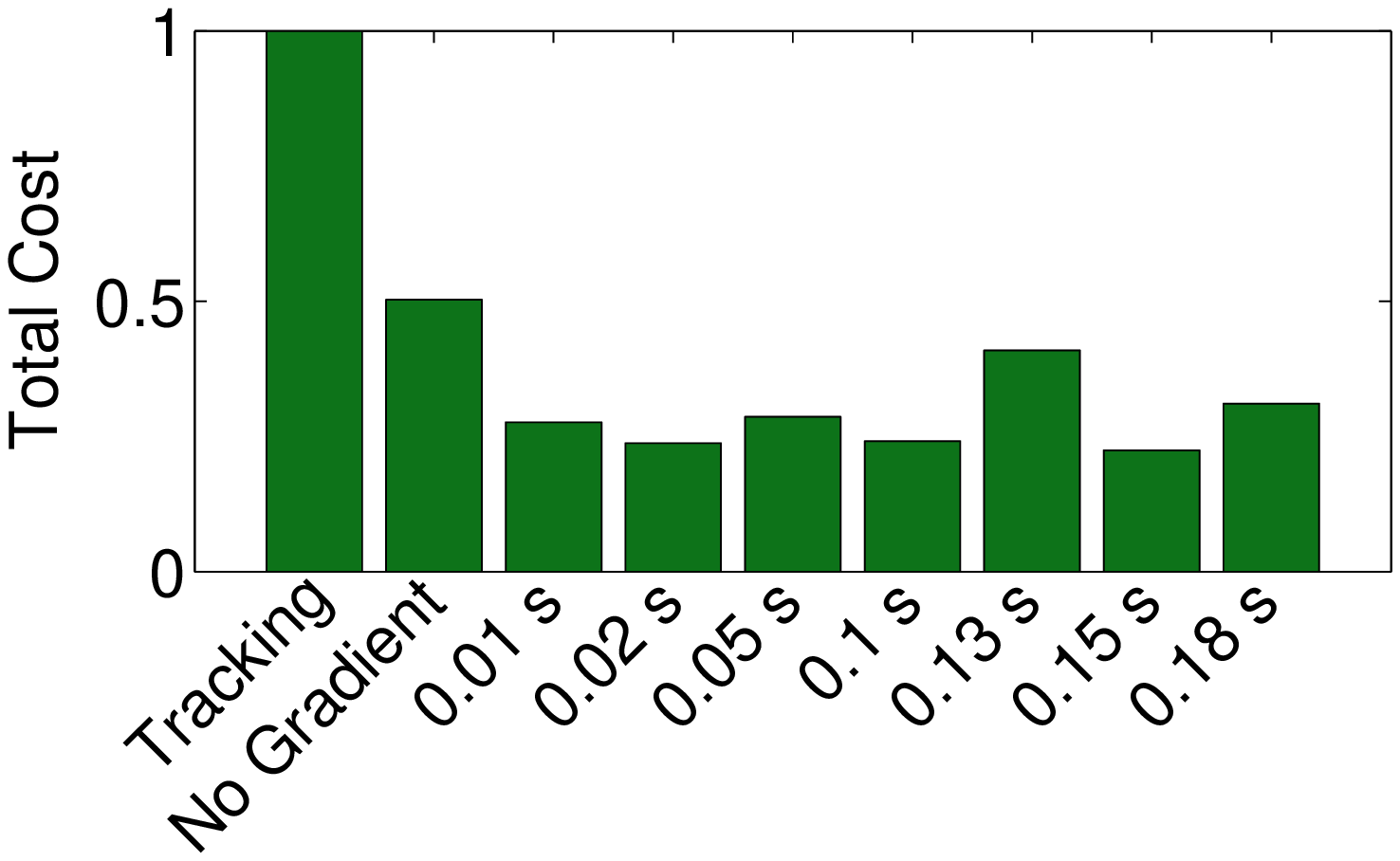} &
\includegraphics[trim=.5cm 1cm 6cm 0cm, clip=true, width=.3\columnwidth]{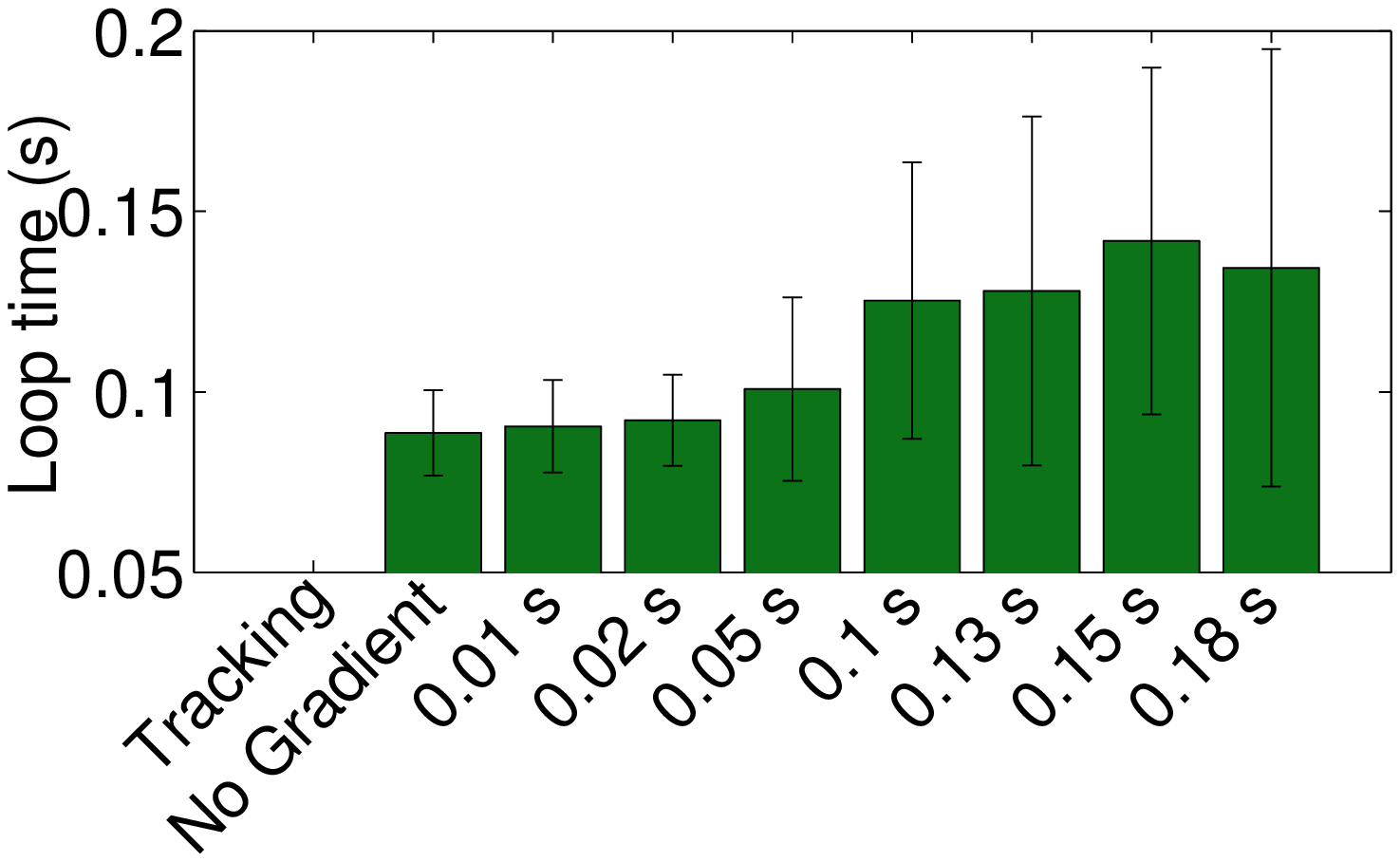} &
\includegraphics[trim=.5cm 1cm 6cm 0cm, clip=true, width=.3\columnwidth]{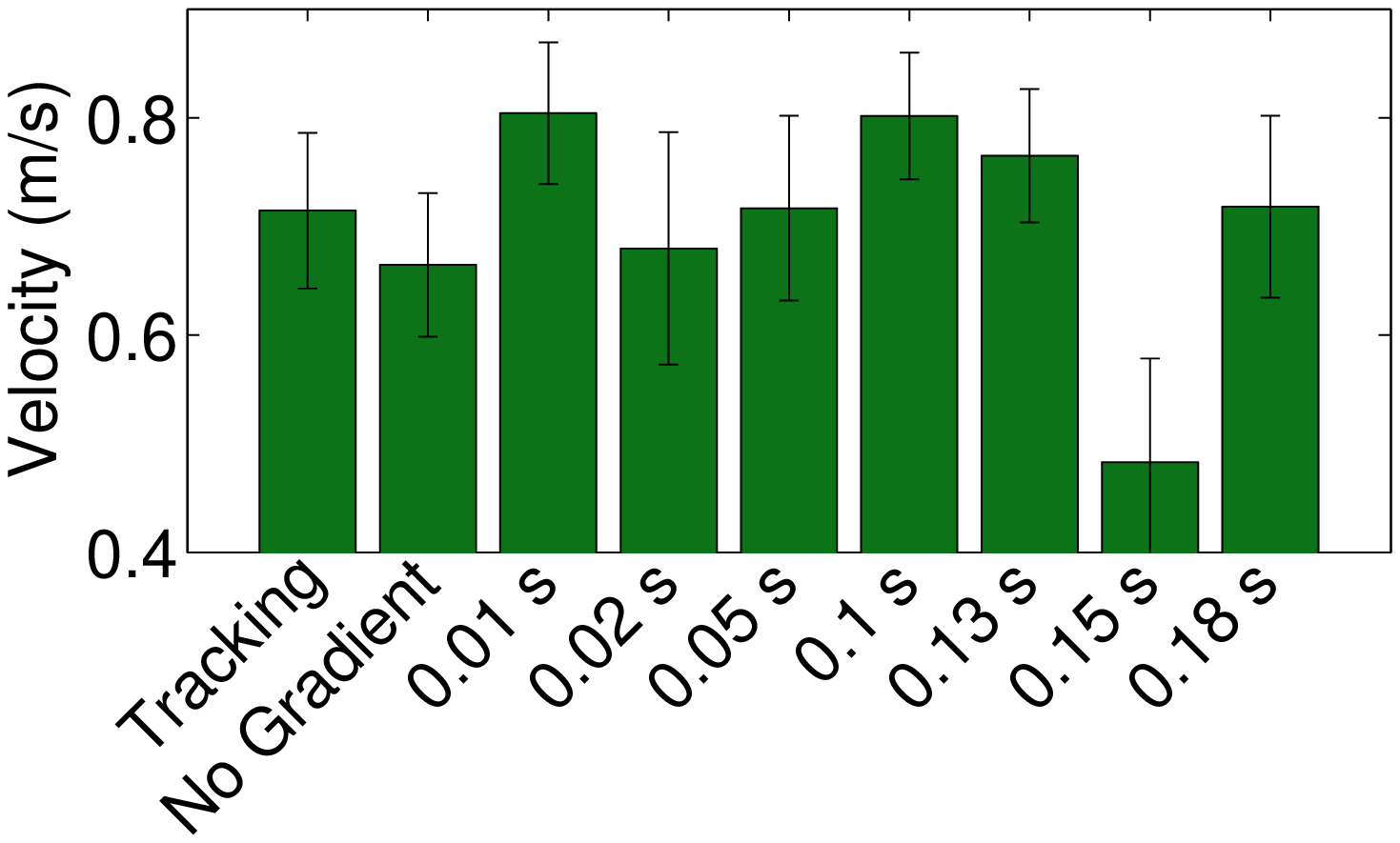} 
\end{array}$
\caption{This figure shows results from the reference tracking control, optimization without gradient descent, and various times allowed for gradient-based optimization.  From left to right is shown the total cost (normalized so that the reference tracking cost equals one), average execution time for a loop of Algorithm \ref{alg:MPCAlg}, and the average velocity for a window time excluding acceleration and deceleration periods.  On both the middle and right plots is shown the standard deviation for each trial. } 
\label{fig:Mag_results}
\end{figure*}

Results to the different trials are shown in Figure \ref{fig:Mag_results}.  Using solely the behavior-based portion of the optimization, the cost was reduced by 50\% when compared to the reference tracking control, with up to a 75\% reduction in cost by including the gradient-based optimization.  It can also be seen that there was no loss in average velocity when comparing the trajectory tracker with the arc-based MPC results.  The highest average velocity seen being 92\% of the desired $.9 \frac{m}{s}$, despite the corridors being less than twice the width of the robot.

While one may expect to see the resulting overall cost to monotonically decrease with allowed optimization, this is not always the case.  As the allowed gradient-descent time approaches $\delta_{execute} = .2$ seconds, the actual time to execute an iteration of Algorithm \ref{alg:MPCAlg} increasingly exceeds $\delta_{execute}$, causing undesirable results.  Also, while it is true that the cost at each time instant will decrease with increased optimization time, when considering the aggregate cost, this need not be the case.  As the information available to the robot is limited, what may seem good at one point in time may prove to have been a poor choice at a future time.  Thus, the costs shown in Figure \ref{fig:Mag_results} are those of local minima and cannot be truly compared beyond the point of noticing the trend that small time allowed for optimization appears to yield good results.

\section{Case Study: Inverted Pendulum Robot}
\label{sec:Segway}
We now examine the problem of performing dual-mode arc-based MPC on a simulation of an inverted pendulum robot.  This provides an example where consideration of the dynamics of the system is very important when planning for the action.  Not only must the planner consider the nonholonomic constraints, it must also maintain balance.  The arc-based MPC algorithm is applied to account for the complicated dynamics while maintaining stability and convergence guarantees.  The section proceeds by first presenting a model, designing control laws, analyzing stability, and ends by giving simulation results.

\subsection{Dynamics of Two-Wheel Inverted Pendulum}
We utilize the model developed in \cite{Kim2005} for an inverted pendulum robot.  The state vector describing the inverted pendulum robot can be expressed as 
\begin{equation}
x = \begin{bmatrix}
x_1 & x_2 & v & \psi & \omega & \phi & \dot{\phi} \end{bmatrix}^T,
\label{eq:segwayStates}
\end{equation}  
where $x_1$, $x_2$, $v$, and $\psi$ are defined as before and $\phi$ is the tilt angle from the vertical, as depicted in Figure \ref{fig:segway_diagram}. This allows the dynamics of the system to be expressed as
\begin{equation}
	\dot{x} = f(x,u) = \begin{bmatrix}
	v \cos(\psi) & 
	v \sin(\psi) &
	\dot{v} &
	\omega &
	\ddot{\psi} &
	\dot{\phi} &
	\ddot{\phi}	
	\end{bmatrix}^T
\label{eq:SegwayDynamics}
\end{equation}
where $\dot{v}$, $\ddot{\psi}$, and $\ddot{\phi}$ are obtained from the following equations: 
\begin{equation}
3(m_c + m_s)\dot{v} - m_s d \cos(\phi) \ddot{\phi} + m_s d \sin(\phi)(\dot{\phi}^2 + \omega^2) = - \frac{1}{R}(\alpha + \beta),
\end{equation}
\begin{equation}
\bigl ( (3 L^2 + \frac{1}{2R^2})m_c + m_s d^2 \sin^2(\phi) + I_2 \bigr ) \ddot{\psi} + m_s d^2 \sin(\phi)\cos(\phi) \omega \dot{\phi} = \frac{L}{R}(\alpha - \beta),
\end{equation}
\begin{equation}
m_s d \cos(\phi) \dot{v} + (-m_s d^2 - I_3) \ddot{\phi} + m_s d^2 \sin(\phi) \cos(\phi) \dot{\phi}^2 + m_s g d \sin(\phi) = \alpha + \beta,
\end{equation}
and the symbols are defined in Table \ref{tab:Symbols}.  Note that, in the simulations, we used the values for the parameters given in \cite{Kim2005}.

\subsection{Control Laws}
\label{ssec:Segway_ctrl}
Two control laws must be defined to apply the dual-mode arc-based MPC algorithm.  The first being a controller to regulate the robot to constant velocities and the second being the reference following controller, which executes time-varying velocities.  Similar to Section \ref{ssec:Mag_ctrl}, we simply mention the point at which a constant velocity control law is straight forward, with the majority of the details focused on the reference tracking control.  

\begin{figure}[!t]
\centering
\includegraphics[width=.25\linewidth]{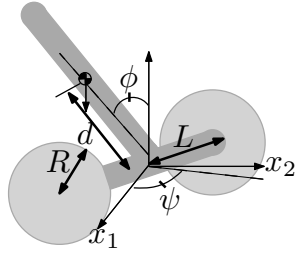}
\caption{Shown is a diagram of the inverted pendulum robot with the symbols defined in Table \ref{tab:Symbols}.}%
\label{fig:segway_diagram}
\end{figure}

\begin{table}
\begin{center}
\caption{This table defines the symbols used in the dynamics of the two-wheel inverted pendulum robot.  The numeric values are given in \protect\cite{Kim2005}. }
\label{tab:Symbols}
		\begin{tabular}{c c c}
		\hline 
		\multicolumn{3}{c}{Table of Symbols}  \\
		\hline
		$m_c$ &\vline & Mass of wheel \\ 
		$m_s$ &\vline & Mass of body \\ 
		$d$ &\vline & Distance from center of wheel axis to center of gravity \\ 
		$L$ &\vline & Half the distance between the wheels \\ 
		$R$ &\vline & Radius of wheels \\ 
		$I_2$  &\vline & Rotational inertia of the body about the $\psi$ axis \\ 
		$I_3$ &\vline & Rotational inertia of the body about the axel \\ 
		$\alpha$. $\beta$ &\vline & Wheel Torques \\ 
		\hline
		\end{tabular}
\end{center} 
\end{table}

To create a velocity-based controller, we use a subset of the states that do not depend on the omitted states, which, when linearized about $x = 0$ and $u = 0$,  produces a system in which the velocities are completely controllable (see \cite{Brogan1974} for details on controllability of linear systems and linearization).  Namely, let $z$ be defined as $z = \begin{bmatrix}
v & \omega & \phi & \dot{\phi} \end{bmatrix}^T$, the following linearized dynamics are obtained:

\begin{equation}
\dot{z} = \begin{bmatrix}
0 & 0 & a_{13} & 0 \\ 
0 & 0 & 0 & 0 \\ 
0 & 0 & 0 & 1 \\ 
0 & 0 & a_{43} & 0
\end{bmatrix} z + 
\begin{bmatrix}
b_{11} & 0 \\
0 & b_{21} \\
0 & 0 \\
b_{41} & 0
\end{bmatrix} u
\label{eq:zLinearized}
\end{equation} 
where
$$
a_{13} = \frac{d^2 g m_s^2}{2 d^2 m_s^2 + 3 m_c d^2 m_s + 3 I_3 m_s + 3 I_3 m_c} = 2.1639,
$$
$$
a_{43} = \frac{3 d g m_s^2 + 3 d g m_c m_s}{(2 d^2 m_s^2 + 3 m_c d^2 m_s + 3 I_3 m_s + 3 I_3 m_c)} = 72.4858,
$$
$$
b_{11} =  -\frac{m_s d^2 + R m_s d + I_3}{R (2 d^2 m_s^2 + 3 m_c d^2 m_s + 3 I_3 m_s + 3 I_3 m_c)} = -1.6687,
$$
$$
b_{21} = \frac{2 L R}{6 m_c L^2 R^2 + 2 I_2 R^2 + m_c} = 0.0290,
$$
$$
b_{41} = -\frac{3 R m_c + 3 R m_s + d m_s}{R (2 d^2 m_s^2 + 3 m_c d^2 m_s + 3 I_3 m_s + 3 I_3 m_c)} = -24.1514,
$$
and 
\begin{equation}
u_1 = (\alpha + \beta), \mbox{ } u_2 = (\alpha - \beta).
\end{equation}

From this linearization, it is seen that $z$ can be further divided into states to control the translational velocity, $z_v = [v, \mbox{ }\phi, \mbox{ }\dot{\phi}]^T$ and the rotational velocity, $\omega$.  From this, it is straight forward to design a control law to execute constant velocities.  Note that similar constructions were given in \cite{Kim2005,Nawawi2006} without recognizing that a linear combination of the control inputs can be used to control each velocity independently.

To control a time varying velocity, a potential problem arises as the same input must be used to both stabilize the pendulum and control the velocity.  To overcome this, we control the tilt angle, which can in turn be used to control the velocity.  This is made possible due to the fact that the tilt angle is much more responsive and that we can define the reference trajectory to maintain nearly constant translational acceleration.  It is also worth noting that it is much easier to maintain stability of the pendulum by controlling the tilt angle to a desired value without any feedback on the velocity.  The feedback on the velocity will come in the form of adjusting the desired tilt angle.

As controlling the desired tilt angle directly affects the translational acceleration, we slightly modify the approach for controlling $y_\epsilon$ given in Section \ref{sec:Unicycle}.  Instead of controlling $\dot{y}$, $\ddot{y}$ is controlled as $\ddot{y} = u_\epsilon$.  This corresponds to the formulation for control presented in \cite{Olfati2002}, namely 
\begin{equation}
\label{eq:Seg_eps_cntrl}
u_\epsilon(t) = \dot{y}_{d} + k_p(y_{d} - y) + k_d (\dot{y}_{d} - \dot{y}_{\epsilon}),
\end{equation}
where $k_d$ and $k_p$ are constants.  Similar to (\ref{eq:Mag_vbar}), the desired accelerations, $\dot{\bar{v}}_d = [\dot{v}_d \mbox{ } \dot{\omega}_d]^T$, can be written as:
\begin{equation}
\dot{\bar{v}}_d = \begin{bmatrix}
 	-\omega \sin(\psi) & \omega \cos(\psi) \\
 	-\frac{\omega \cos(\psi)}{\epsilon} &  	-\frac{\omega \sin(\psi)}{\epsilon}
 	\end{bmatrix} \dot{y} 
 	+ \begin{bmatrix}
 	\cos(\psi) & \sin(\psi) \\
	-\frac{\sin(\psi)}{\epsilon} & \frac{\cos(\psi)}{\epsilon}
 	\end{bmatrix} u_\epsilon.
 	\label{eq:Seg_vbar}
\end{equation}
Given the desired accelerations, we allow the inverted pendulum control to take the form
\begin{equation}
\label{eq:Seg_ref_ctrl}
\begin{split}
	u_1 &= -k_\phi e_\phi - \frac{a_{43}}{b_{41}}\phi_d \\
	u_2 &= \frac{\dot{\omega}_d}{b_{21}}
\end{split}
\end{equation}
where
$$
	\phi_d = \frac{b_{41} a_{13}}{a_{13}^2 b_{41} - b_{11} a_{43}} \dot{v}_d,
$$
$$
	e_\phi = \begin{bmatrix}
	\phi - \phi_d \\
	\dot{\phi}
	\end{bmatrix}
$$
and $k_\phi$ is a feedback matrix.  Assuming constant acceleration, the dynamics of the error for the linearized system becomes
$$
	\dot{e}_\phi = (\begin{bmatrix}
	0 & 1 \\
	a_{43} & 0
	\end{bmatrix} - 
	\begin{bmatrix}
	0 \\
	b_{41}
	\end{bmatrix} k_\phi) e_\phi,	
$$
which has exponential convergence rates to the desired tilt angle.  Similar to the development in Section \ref{ssec:Mag_ctrl}, Assumptions \ref{as:limitedAcc_RefTraj} and \ref{as:proximit_vel_Mag} can be made to ensure the convergence to the desired reference trajectory.  Similarly, to ensure that the control law can converge to the necessary characteristics within $\delta_{execute}$ seconds, a cost barrier could be introduced.  However, in the same fashion as mentioned in Section \ref{ssec:Mag_ctrl}, we found that with the initialization step, this was not needed in practice.

\subsection{Maintaining Balance}
To ensure that the pendulum robot maintains balance, we analyze both the stability of the control law presented in the previous section as well as the ability for Algorithm \ref{alg:MPCAlg} to maintain balance.  While the previous section considered the linearized dynamics, this section considers the balancing on the full, nonlinear dynamics.

Under an assumption that $\dot{v}_d$ is constant, we can utilize the Lyapunov function $V = \frac{1}{2} e_\phi ^T P e_\phi$ and evaluate numerically the space such that $\dot{V} = e_\phi ^T P \dot{e}_\phi < 0$.  Using the feedback gain matrix $k_\phi = [-8.5223 \mbox{ } -.9922]$ and 
$$
P = \begin{bmatrix}
2.8929 & 0.0037 \\
0.0037 & .0210
\end{bmatrix},
$$
we found that $\dot{V} < 0$ $\forall \mbox{ }e_\phi \neq 0$ for $|\phi| \leq .5$, $|\dot{\phi}| < 5$, and extreme initial conditions(i.e. values greater than we ever observed) of $v = 2.0$ and $\dot{\omega} = 2$ for all $|\phi_d| \leq 0.35$.  To ensure stability, we place a bound on the desired tilt angle such that $|\phi_d| < .25$.

To ensure that the robot will maintain balance while executing Algorithm \ref{alg:MPCAlg}, a barrier cost is introduced as part of the instantaneous cost.  The barrier cost takes the form:
\begin{equation}
\label{eq:Seg_phi_cost}
-\rho_6 \log \bigl(\frac{\phi_{max}^2 - \phi(t)^2}{\phi_{max}^2} \bigr),
\end{equation}
where $\phi_{max} = .5$ is the maximum allowable tilt angle and $\rho_6$ is a constant.  This barrier cost will not permit step \ref{MPCAlg:Optimize} of Algorithm \ref{alg:MPCAlg} to produce a solution with a resulting $\phi(t) > \phi_{max}$ for some $t \in [t_0, t_0 + \Delta]$ if it is initialized with a solution that satisfies the constraint.  Moreover, the final control given in (\ref{eq:Seg_ref_ctrl}) will present an admissible solution with $\phi(t) < \phi_{max}$ $\forall t \in [t_0, t_0 + \Delta]$ at the next iteration of the algorithm.  All other costs discussed in Section \ref{ssec:UnicycleCosts} remain the same.

\subsection{Results}

We utilize the inverted pendulum robot to demonstrate the ability of dual-mode arc-based MPC to take a complicated dynamic model into account and ensure stability.  Results were obtained on a simulation of the inverted pendulum robot through the environment shown in Figure \ref{fig:Seg_environment} using a dual-core i7 2.67 GHz processor.  The mapping and visualization was again performed using ROS, with a simulated laser range finder giving data to ROS to form the map.  As before, several trials were performed to evaluate the performance of the dual-mode arc-based MPC algorithm.  

\begin{figure*}[!t]
\centering
$
\begin{array}{cccc}
\includegraphics[trim=2cm 2cm 1cm 2cm, clip=true, width=.22\columnwidth]{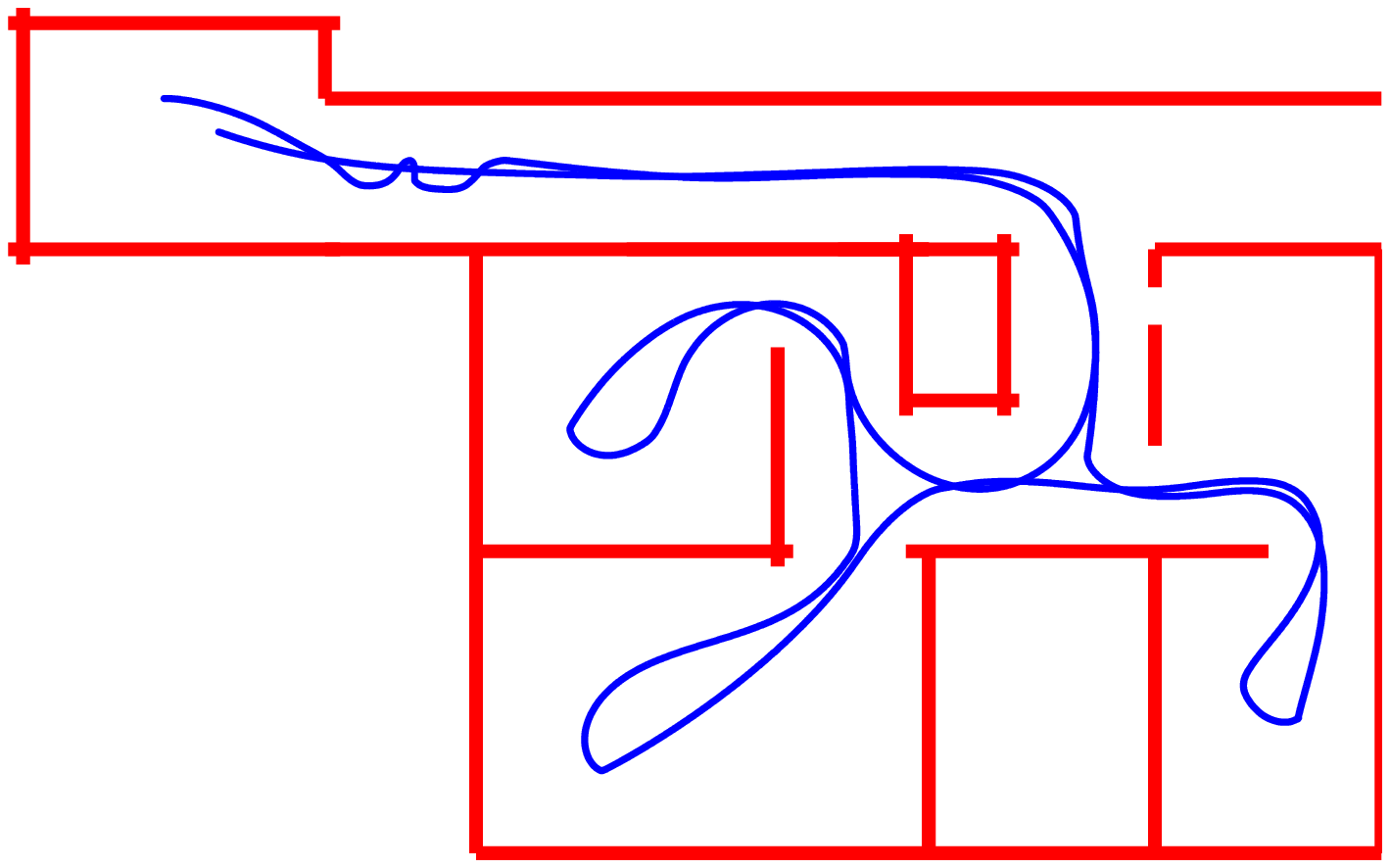} &
\includegraphics[trim=12cm 14cm 6cm 0cm, clip=true, width=.22\columnwidth]{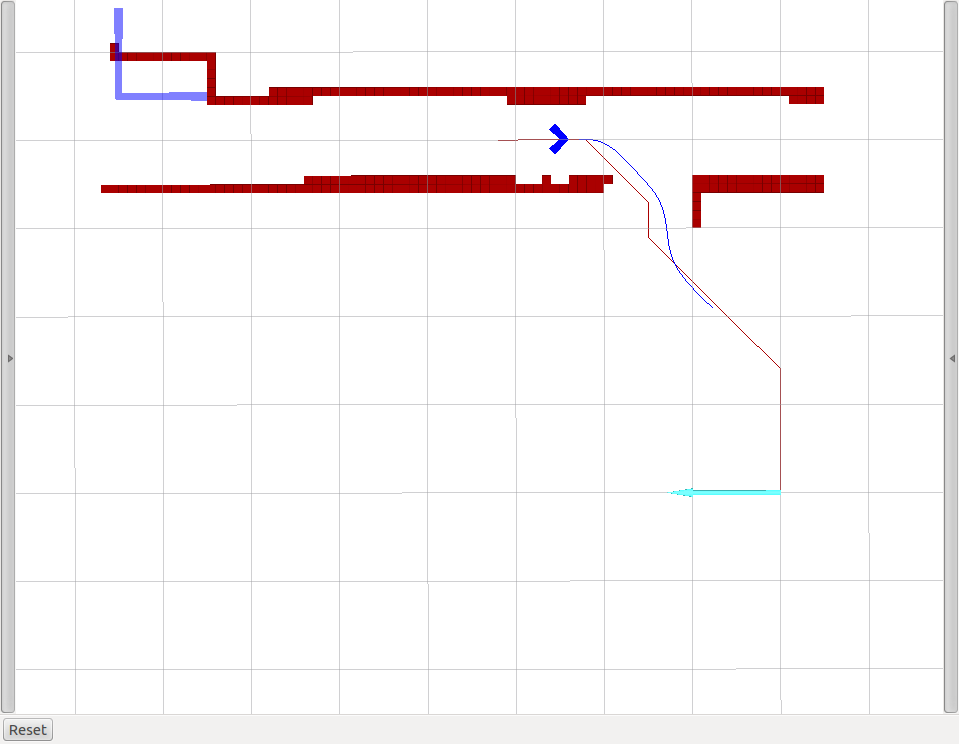} &
\includegraphics[trim=16cm 8cm 2cm 6cm, clip=true, width=.22\columnwidth]{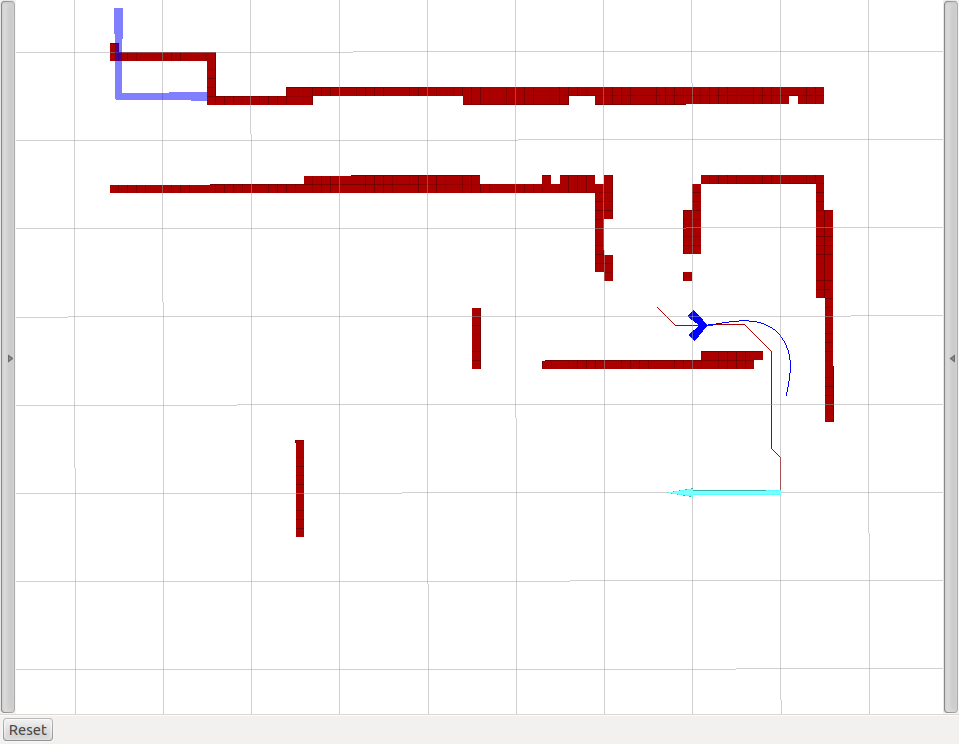} &
\includegraphics[trim=9cm 8cm 9cm 6cm, clip=true, width=.22\columnwidth]{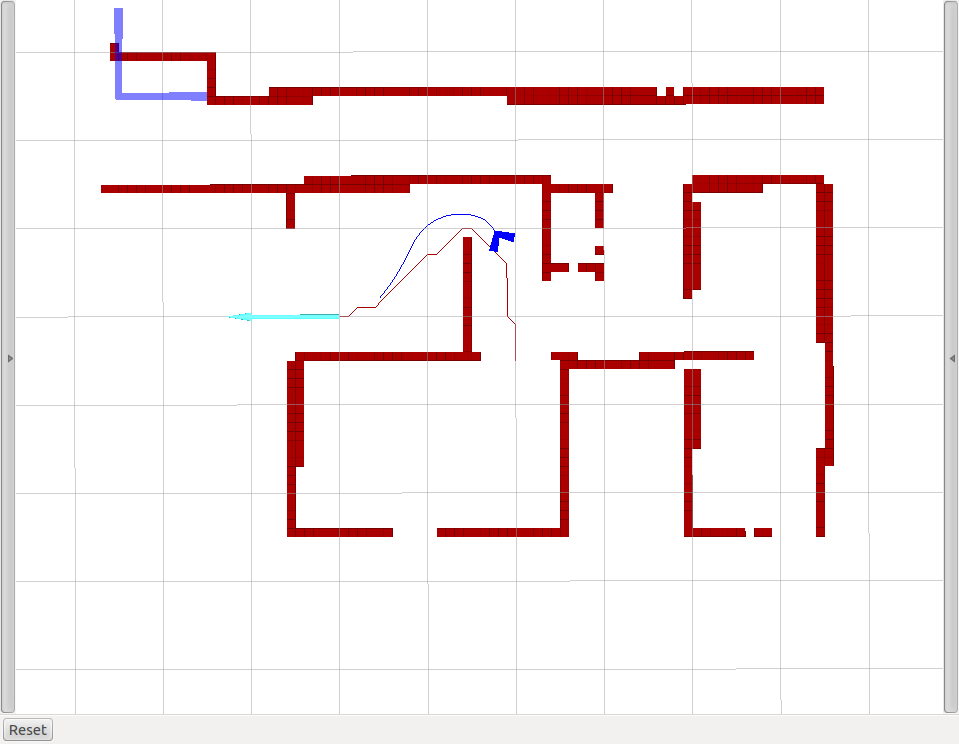} 
\end{array}$
\caption{On the left is shown an executed path through the environment.  The right three images show the robot and current map at different positions along the trajectory.  The robot is shown as a triangle.  The line extending from it is the trajectory created by the arc-based MPC framework.}
\label{fig:Seg_environment}
\end{figure*}

\begin{figure*}[!t]
\centering
$
\begin{array}{ccc}
\includegraphics[trim=.5cm 1cm 6cm 0cm, clip=true, width=.3\columnwidth]{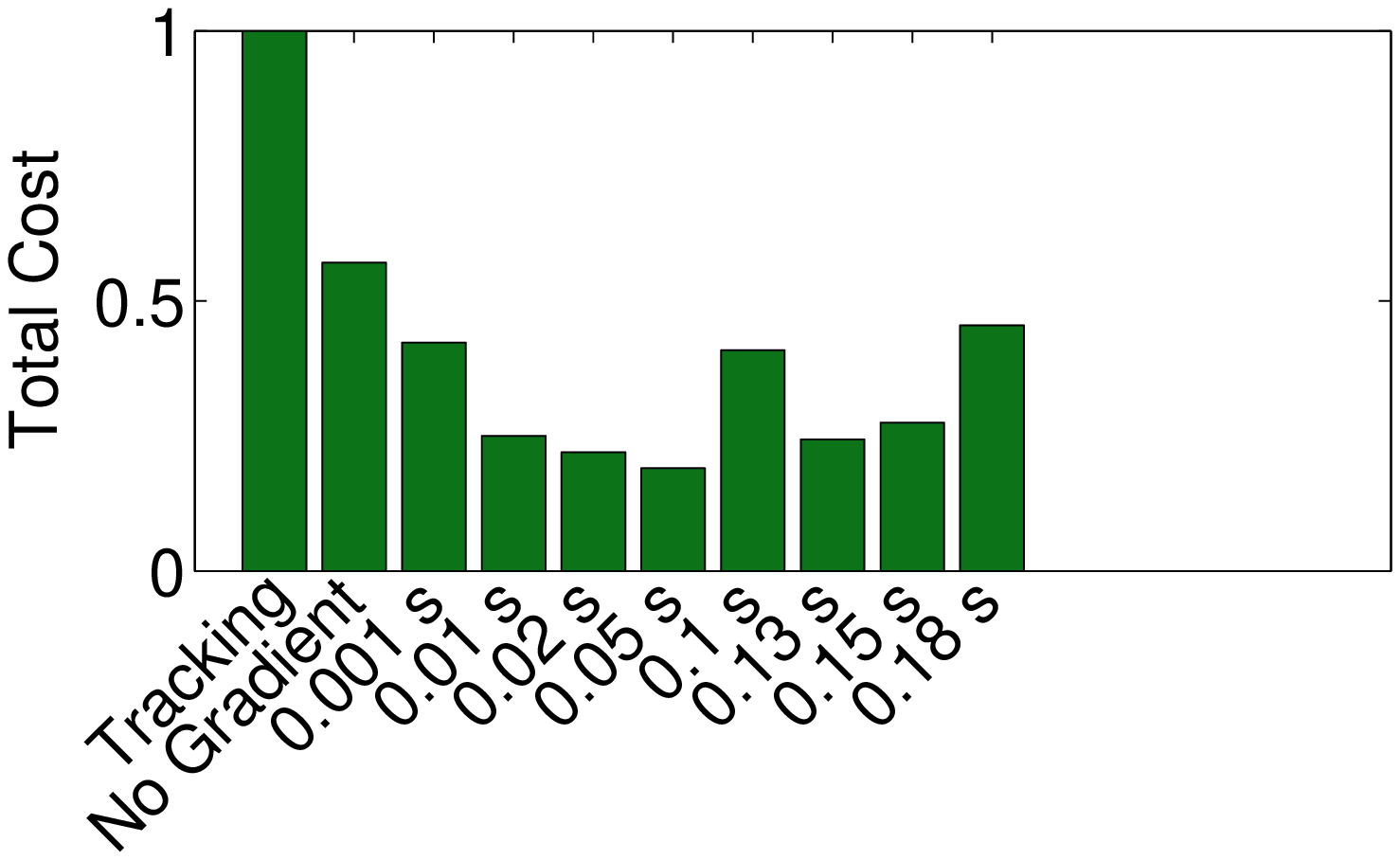} &
\includegraphics[trim=.5cm 1cm 6cm 0cm, clip=true, width=.3\columnwidth]{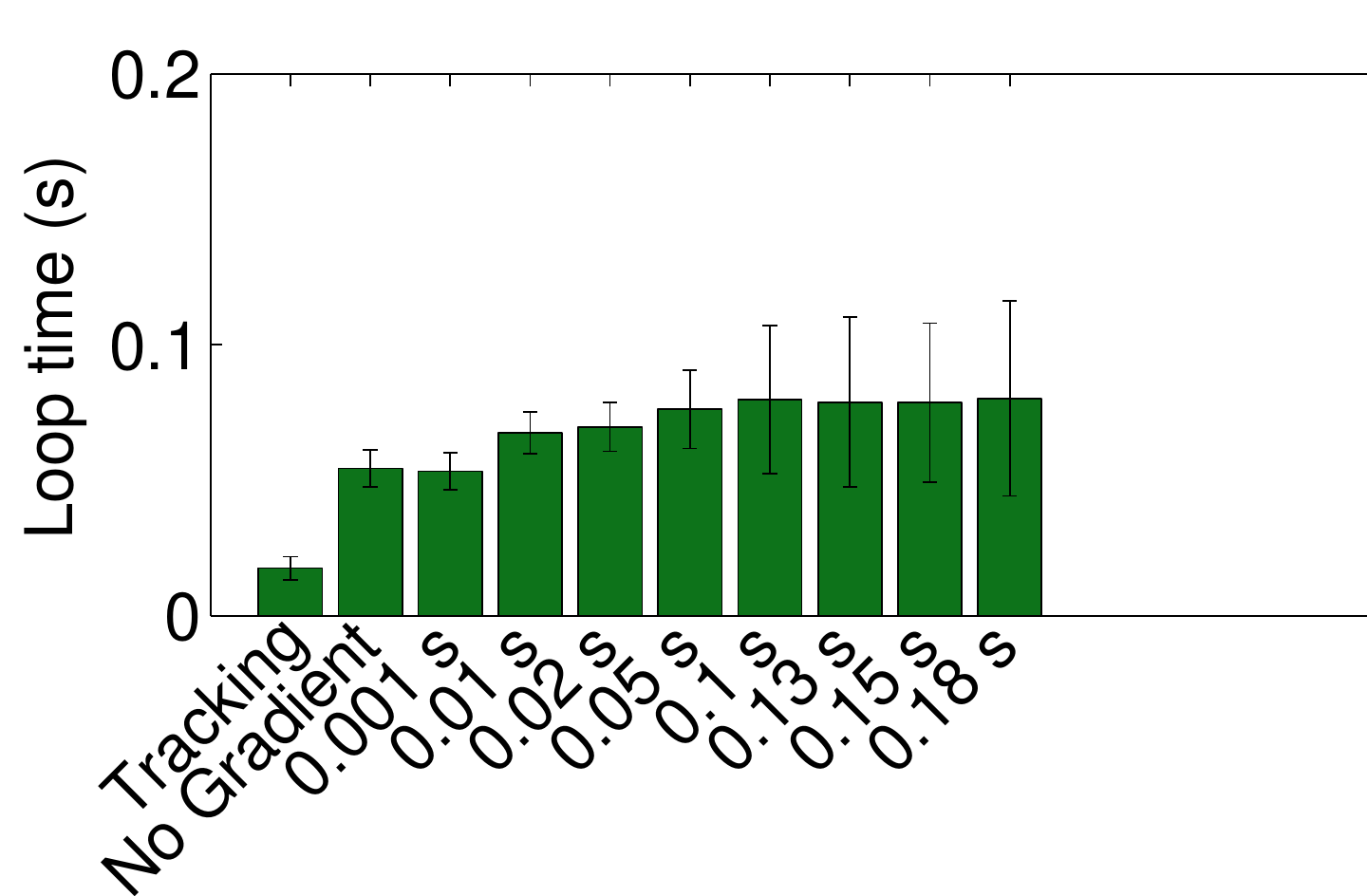} &
\includegraphics[trim=.5cm 1cm 6cm 0cm, clip=true, width=.3\columnwidth]{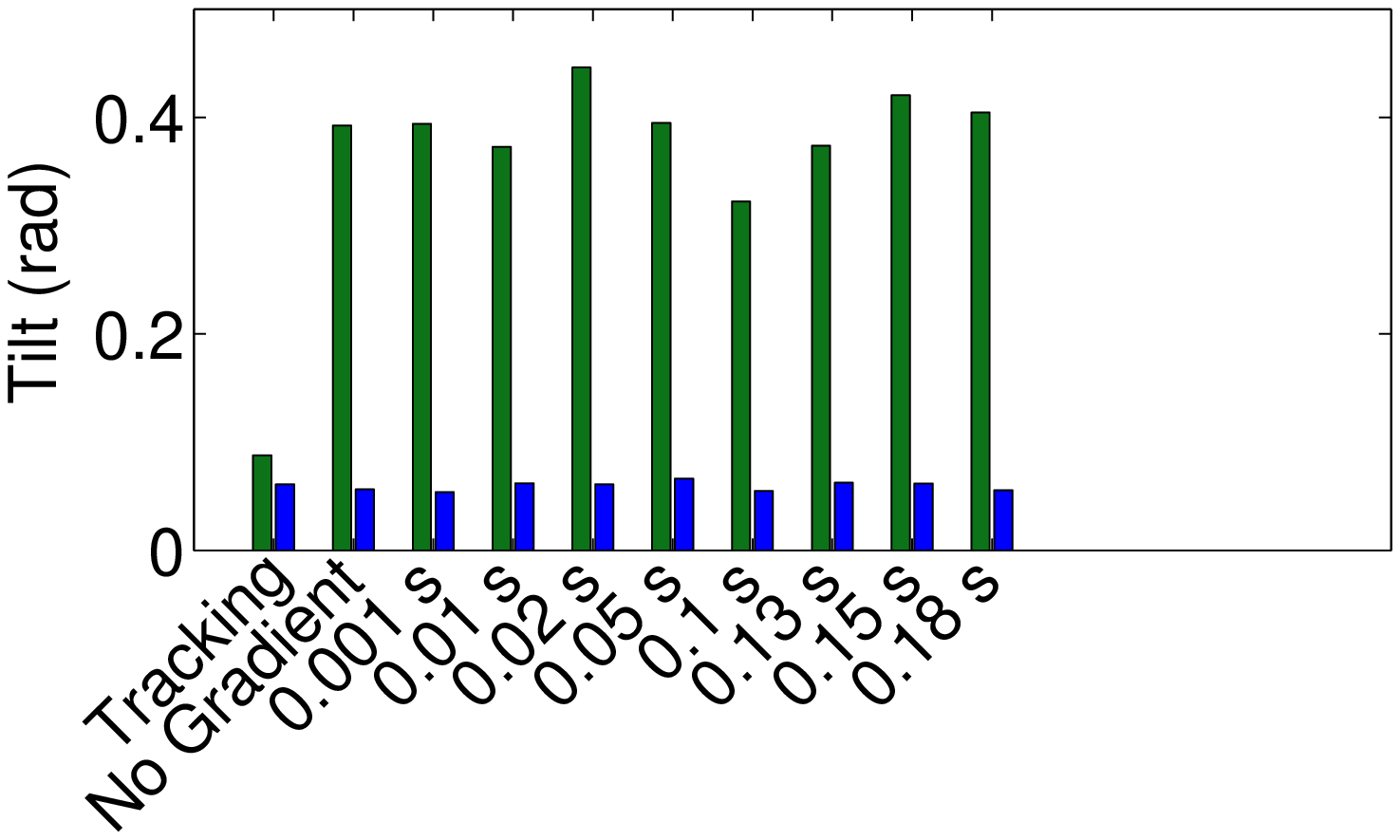} 
\end{array}$
\caption{This figure shows results for the inverted pendulum robot using the reference tracking control, optimization without gradient descent, and various times allowed for gradient-based optimization.  From left to right is shown the total cost (normalized so that the reference tracking cost equals one), average execution time for a loop of Algorithm \ref{alg:MPCAlg}, and average and maximum tilt angle for each trial.} 
\label{fig:Seg_results}
\end{figure*}

Each trial consisted of navigating to a series of waypoints throughout the environment, with an $A^*$ path planner being used to plan two dimensional paths between waypoints.  The robot was given a desired speed of $1 \frac{m}{s}$, with each trial having the robot approach that speed between waypoints.   The trials included using the reference tracking control, arc-based MPC without gradient descent, and using several limits on the allowed time for gradient descent.  Results are shown in Figure \ref{fig:Seg_results} and Extension 1.

Several trends can be observed on the results.  Similar to before, the behavior-based portion of the optimization significantly reduces the cost when compared to the reference tracking control.  The inclusion of a gradient-based optimization step provides even better results.  By observing the loop time versus time allotted for gradient descent, it is apparent that convergence of the optimization at each time step is achieved rather quickly despite the complicated dynamic model.  The average time to execute a loop of Algorithm \ref{alg:MPCAlg} is less than $\frac{1}{2} \delta_{execute}$, even as the alloted time approaches $\delta_{execute}$.  Most important to the development of this section, the robot maintained balance.  The maximum tilt angle in each trial is below the tilt constraint, with the average being much smaller.

\section{Conclusion}
\label{sec:Conclusion}
In this paper we have developed a dual-mode arc-based MPC algorithm which ensures obstacle avoidance and guarantees convergence to a desired goal location despite complicated dynamic models.  Dynamic motion constraints, limited accelerations, and stability concerns were considered in the examples without imposing unreasonable computational demands.  The algorithm was applied to a decade old Irobot Magellan-pro, which maintained an average of 80\% of its top speed while navigating through corridors less than twice the robot's width without collision.  The ability of the algorithm to deal with more complicated dynamics, including stability concerns, was illustrated through the example of an inverted pendulum robot where it reached high speeds while traversing a complicated environment without exceeding a pre-specified maximum tilt angle.

\section*{Funding}
\small{This research received no specific grant from any funding agency in the public, commercial, or not-for-profit sectors.}

\bibliographystyle{apacite}
\bibliography{Submission}

\end{document}